\documentclass[letterpaper,12 pt,onecolumn]{ieeeconf}
\IEEEoverridecommandlockouts 
\usepackage{amsmath,amssymb,amsfonts}
\usepackage{graphicx}
\usepackage{color}
\usepackage{bbding}
\usepackage{pifont}
\usepackage{subfigure}
\usepackage{balance}
\usepackage{hyperref}
\usepackage{textcomp}
\usepackage[T1]{fontenc}
\usepackage{cite}
\usepackage{hyperref}
\usepackage{multicol,lipsum}
\usepackage[justification=centering]{caption}
\usepackage{graphicx}
\usepackage{algorithm}
\usepackage{algorithmic}
\usepackage{multirow}
\usepackage{hyperref}
\usepackage{mathtools}

\usepackage{bm}
\usepackage{listings}
\newtheorem{remark}{Remark}
\newtheorem{definition}{Definition}
\newtheorem{theorem}{Theorem}
\newtheorem{proposition}{Proposition}
\newtheorem{lemma}{Lemma}
\newtheorem{example}{Example}
\newtheorem{corollary}{Corollary}
\newtheorem{assumption}{Assumption}

\newcommand{\oomit}[1]{}

\begin{document}
\title{Sufficient and Necessary Barrier-like Conditions for Safety and Reach-avoid Verification of Stochastic Discrete-time Systems
}
\author{Bai Xue
 \thanks{Key Laboratory of System Software (Chinese Academy of Sciences) and State Key Laboratory of Computer Science, Institute of Software, Chinese Academy of Sciences, Beijing, China (xuebai@ios.ac.cn).}
}

\maketitle
\thispagestyle{empty}
\pagestyle{empty}

\begin{abstract}
This paper investigates necessary and sufficient barrier-like conditions for infinite-horizon safety and reach-avoid verification of stochastic discrete-time systems, derived via a relaxation of the Bellman equations. Unlike prior approaches that primarily focus on sufficient conditions, our work rigorously establishes both necessity and sufficiency for infinite-horizon properties. Safety verification concerns certifying that, starting from a given initial state, the system remains within a safe set at all future time steps with probability at least equal to a specified threshold. For this purpose, we formulate a necessary and sufficient barrier-like condition that captures this infinite-time safety property. In contrast, reach-avoid verification generalizes safety verification by also incorporating reachability. Specifically, it aims to ensure that the probability of the system, starting from a given initial state, eventually reaching a target set while remaining within the safe set until the first hit of the target is no less than a prescribed bound. Under suitable assumptions, we establish two necessary and sufficient barrier-like conditions for this reach-avoid specification.
\end{abstract}

\section{Introduction}
\label{intro}
Temporal verification is crucial in modern systems analysis, particularly in complex systems where temporal behavior is of paramount importance \cite{prajna2007convex}. It involves rigorously examining a system's adherence to temporal properties, including safety and reach-avoid guarantees, to ensure desired outcomes and avoid undesirable events. Formal methods like model checking \cite{clarke1997model} and theorem proving \cite{manna2012temporal} are indispensable tools in this process, allowing for precise and comprehensive analysis of temporal specifications.

Over the past two decades, barrier certificates have become a powerful tool for safety and reach-avoid verification of dynamical systems. These certificates provide Lyapunov-like guarantees regarding system behavior. The existence of a barrier certificate alone is sufficient to establish the satisfiability of safety and reach-avoid specifications, as demonstrated in \cite{prajna2007convex}. This simplifies the verification process and provides a formal mathematical framework for ensuring the safety and correctness of a system without needing to explicitly evolve it over time. With advances in polynomial optimization, particularly sum-of-squares polynomial optimization, barrier certificates can be computed through convex optimization, especially when the system of interest is polynomial. This further motivates the development of barrier certificate-based methods.

On the other hand, converse theorems for barrier certificates, which focus on the existence of such certificates, have significantly contributed to understanding how safety and reach-avoid criteria can be represented by barrier certificates. These concepts have garnered growing interest since the inception of barrier certificates and have been further investigated in \cite{prajna2005necessity, prajna2007convex, wisniewski2015converse, ratschan2018converse, liu2021converse}. However, there remains a scarcity of research exploring the existence of barrier certificates for stochastic dynamical systems. This work aims to fill this gap.

By relaxing Bellman equations, this paper derives necessary and sufficient barrier-like conditions for verifying safety and reach-avoid properties in stochastic discrete-time systems over infinite-time horizons. The safety verification process involves assessing whether the safety probability that a system, starting from an initial state, will stay within a safe set for all time is greater than or equal to a specified threshold. By relaxing a Bellman equation, one of whose solutions characterizes the exact safety probability, we construct a necessary and sufficient  barrier-like condition for safety verification. On the other hand, the reach-avoid verification concerns verifying whether the reach-avoid probability that the system, starting from an initial state, will enter a target set eventually while avoiding unsafe sets before hitting the target, is greater than or equal to a specified threshold. We consider two cases for the reach-avoid verification. In the first case, we assume that, for every state in the safe set but not in the target set, the system will almost surely either reach the target set or exit the safe set in finite time. Under this context, by relaxing a Bellman equation, which possesses a unique bounded solution that characterizes the exact reach-avoid probability, we construct a necessary and sufficient  barrier-like condition for the reach-avoid verification. In the second case, we assume that the specified threshold is strictly smaller than the exact reach-avoid probability. Under this context, by relaxing a Bellman equation featuring a unique bounded solution that provides a lower bound of the exact reach-avoid probability, we construct a necessary and sufficient  barrier-like condition for the reach-avoid verification.

\section*{Related Work}

Barrier certificates were initially proposed for deterministic systems as a formal approach to safety verification \cite{prajna2004safety}. Subsequent efforts have focused on adapting and enhancing these functions, broadening their applications \cite{kong2013exponential, ames2019control, ghaffari2018safety, anand2021safety}. However, many real-world applications are subject to stochastic disturbances, prompting the modeling of these systems as stochastic. In the continuous-time stochastic setting, safety verification over the infinite-time horizon using barrier certificates was introduced alongside its deterministic counterpart in \cite{prajna2007framework}. Based on Ville's Inequality \cite{ville1939etude} and a stopped process, \cite{prajna2007framework} developed a non-negative barrier function and established a sufficient condition for safety verification, certifying upper bounds on the probabilities of entering an unsafe region from specific initial states. This ensures that the system remains within the interior of a state-constrained set until its first encounter with the unsafe set. Building on \cite{kushner1967stochastic}, \cite{santoyo2021barrier} formulated a sufficient barrier-like condition for upper-bounding the probability of entering an unsafe region from certain initial states within finite-time frames. The systems in \cite{santoyo2021barrier} include both continuous-time and discrete-time systems. Notably, when the state-constrained set is a robust invariant (i.e., systems initialized within it remain within it under all disturbances) and the unsafe set is a subset of the invariant set, sufficient barrier-like conditions for safety verification of stochastic discrete-time systems were studied in \cite{anand2022k, zhi2024unifying}. Another commonly studied safety property is related to set invariance. This involves justifying lower bounds of safety probabilities, either over an infinite time horizon (i.e., ensuring the system remains within a specified safe set for all time) or finite time horizons (i.e., ensuring the system stays within a given safe set during a specified time period) \cite{abate2008probabilistic}. In other words, it involves justifying upper bounds on the exit probabilities—i.e., the probability that the system will eventually exit a specified safe set or do so within a bounded time horizon. To address this, sufficient barrier-like conditions have been developed for safety verification over both finite-time horizons (e.g., \cite{steinhardt2012finite, laurenti2023unifying, cosner2023robust, cosner2024bounding}) and infinite-time horizons (e.g., \cite{yu2023safe}). Following this, control barrier functions were explored for synthesizing controllers to guarantee safety in \cite{sarkar2020high, wang2021safety}. While finite-time verification suffices for systems with bounded operational horizons, we emphasize the importance of infinite-time safety guarantees—essential for systems requiring robustness against indefinite or unpredictable mission lifetimes. Importantly, the proposed method can also be applied to safety verification over an infinite time horizon, as described in \cite{prajna2007framework}.

Regarding reach-avoid verification, a new sufficient barrier-like condition was proposed in \cite{xue2021reach} for the reach-avoid analysis of stochastic discrete-time dynamical systems over an infinite-time horizon. This condition was later extended to stochastic continuous-time dynamical systems in \cite{xue2024}. The condition is constructed by relaxing a set of equations, whose solution characterizes the exact reach-avoid probability of eventually entering a desired target set from an initial state while maintaining safety constraints. In addition, another barrier-like function, called reach-avoid supermartingales, was introduced in \cite{vzikelic2023learning, vzikelic2023compositional} to guarantee reach-avoid specifications and facilitate controller synthesis for stochastic discrete-time systems. This framework assumes that the system evolves within a robust invariant set, with both the unsafe set and target set confined within this invariant domain. However, this strong assumption limits the applicability of the framework, as many systems do not possess compact robust invariant sets, as discussed in \cite{yu2023safe}. These barrier-like conditions aim to lower bound reach-avoid probabilities, as explored in \cite{xue2021reach, vzikelic2023learning, vzikelic2023compositional, xue2024}.

This paper is structured as follows: Section \ref{sec:pre} introduces the stochastic discrete-time systems of interest and formulates the safety and reach-avoid verification problems. Section \ref{sec:rav} presents the necessary and sufficient barrier-like conditions for safety verification. Section \ref{sec:lb} follows with the necessary and sufficient barrier-like conditions for reach-avoid verification. Section \ref{sec:ex} presents two numerical examples that demonstrate the effectiveness of the proposed barrier-like conditions. Finally, Section \ref{sec:con} concludes the paper.

\section{Preliminaries}
\label{sec:pre}
We start the exposition by a formal introduction of stochastic discrete-time systems and safety/reach-avoid verification problems of interest. Before posing the problem studied, let us introduce some basic notions used throughout this paper: $\mathbb{R}$ denotes the set of real values; $\mathbb{N}$ denotes the set of nonnegative integers; $\mathbb{N}_{\leq k}$ is the set of non-negative integers being less than or equal to $k$; $\mathbb{N}_{\geq k}$ is the set of non-negative integers being larger than or equal to $k$; for sets $\Delta_1$ and $\Delta_2$, $\Delta_1\setminus \Delta_2$ denotes the difference of sets $\Delta_1$ and $\Delta_2$, which is the set of all elements in $\Delta_1$ that are not in $\Delta_2$;  $1_A(\bm{x})$ denotes the indicator function in the set $A$,
where, if $\bm{x}\in A$, then $1_A(\bm{x}) = 1$ and if $\bm{x}\notin A$, $1_A(\bm{x}) = 0$.  

\subsection{Problem Statement}

This paper considers stochastic discrete-time systems that are modeled by stochastic difference equations of the following form:
\begin{equation}
\label{system}
\begin{split}
\bm{x}(l+1)=\bm{f}(\bm{x}(l),\bm{\theta}(l)), \forall l\in \mathbb{N},
\end{split}
\end{equation}
where $\bm{x}(l) \in \mathbb{R}^n$ is the state at time $l$ and $\bm{\theta}(l)\in \Theta$ with $\Theta \subseteq \mathbb{R}^m$ is the stochastic disturbance at time $l$. In addition, let $\bm{\theta}(0), \bm{\theta}(1), \ldots$ be i.i.d. (independent and identically distributed) random variables on a probability space $(\Theta,\mathcal{F},\mathbb{P}_{\bm{\theta}})$, and take values in $\Theta$ with the following probability distribution: for any measurable set $B\subseteq \Theta$, 
\[{\rm Prob}(\bm{\theta}(l)\in B)=\mathbb{P}_{\bm{\theta}}(B), \quad\forall l\in \mathbb{N}.\]
The corresponding expectation is denoted as $\mathbb{E}_{\bm{\theta}}[\cdot]$.

Before defining the trajectory of system \eqref{system}, we define a disturbance signal.
\begin{definition}
A disturbance signal $\pi$ is a sample path of the stochastic process ${\boldsymbol{\theta}(i): \Theta \to \Theta, i \in \mathbb{N}}$, defined on the canonical sample space $\Theta^{\infty}$ equipped with the product topology and Borel $\sigma$-algebra $\mathcal{B}(\Theta^{\infty})$. The probability measure $\mathbb{P}{\pi} := \mathbb{P}_{\boldsymbol{\theta}}^{\infty}$ is the product measure on $\Theta^{\infty}$ induced by the i.i.d. disturbances $\boldsymbol{\theta}(0), \boldsymbol{\theta}(1), \dots$: $\mathbb{P}_{\pi}=\bigotimes_{i=0}^{\infty} \mathbb{P}_{\boldsymbol{\theta}}$,
where $\mathbb{P}_{\bm{\theta}}(B)={\rm Prob}(\bm{\theta}(i)\in B)$ for measurable $B\subseteq \Theta$. The expectation $\mathbb{E}_{\pi}[\cdot]$ is defined with respect to $\mathbb{P}_{\pi}$.
\end{definition}

A disturbance signal $\pi$ together with an initial state $\bm{x}_0\in \mathbb{R}^n$ induces a unique discrete-time trajectory as follows.
\begin{definition}
Given a disturbance signal $\pi$ and an initial state $\bm{x}_0\in \mathbb{R}^n$, a trajectory of system \eqref{system} is denoted as  $\bm{\phi}_{\pi}^{\bm{x}_0}(\cdot)\colon\mathbb{N}\rightarrow \mathbb{R}^n$ with $\bm{\phi}_{\pi}^{\bm{x}_0}(0)=\bm{x}_0$, i.e.,
\[\bm{\phi}_{\pi}^{\bm{x}_0}(l+1)=\bm{f}(\bm{\phi}_{\pi}^{\bm{x}_0}(l),\bm{\theta}(l)), \forall l\in \mathbb{N}.\]
\end{definition}

The safety and reach-avoid verification for the system \eqref{system} over the infinite-time horizon are defined below.

\begin{definition}[Safety Verification]
\label{safe}
Given a safe set $\mathcal{X}\subseteq \mathbb{R}^n$, an initial state $\bm{x}_0$, and a lower bound $\epsilon_1 \in [0,1]$, the safety verification aims to certify that the safety probability $\mathbb{P}_{\pi}(S_{\bm{x}_0})$, which denotes the probability that the system \eqref{system}, starting from the initial state $\bm{x}_0$, will stay within the safe set $\mathcal{X}$ for all time, is greater than or equal to $\epsilon_1$, i.e., 
 \[\mathbb{P}_{\pi}(S_{\bm{x}_0})\geq \epsilon_1,\]
 where $S_{\bm{x}_0}=\{\pi \mid \forall i\in \mathbb{N}.  \bm{\phi}_{\pi}^{\bm{x}_0}(i)\in \mathcal{X}\}$.
\end{definition}

\begin{definition}[Reach-avoid Verification]
\label{reach-avoid}
Given a safe set $\mathcal{X}\subseteq \mathbb{R}^n$, an initial state $\bm{x}_0\in \mathcal{X}\setminus \mathcal{X}_r$, a target set $\mathcal{X}_r\subseteq \mathcal{X}$, and a lower bound $\epsilon_2 \in [0,1]$, the reach-avoid verification aims to certify that the reach-avoid probability $\mathbb{P}_{\pi}(RA_{\bm{x}_0})$, which denotes the probability that system \eqref{system}, starting from the initial state $\bm{x}_0$, will reach the target set $\mathcal{X}_r$ eventually while staying within the safe set $\mathcal{X}$, is greater than or equal to $\epsilon_2$, i.e., 
 \[\mathbb{P}_{\pi}(RA_{\bm{x}_0})\geq \epsilon_2,\]
 where $RA_{\bm{x}_0}=\{\pi \mid \exists k\in \mathbb{N}. \bm{\phi}_{\pi}^{\bm{x}_0}(k)\in \mathcal{X}_r \wedge \forall i\in \mathbb{N}_{\leq k}.  \bm{\phi}_{\pi}^{\bm{x}_0}(i)\in \mathcal{X}\}$.
\end{definition}

In the sequel, we will formulate necessary and sufficient barrier-like conditions for certifying $\epsilon_1\leq \mathbb{P}_{\pi}(S_{\bm{x}_0})$. We note that the method can also be used to construct necessary and sufficient conditions for the safety verification scenario in \cite{prajna2007framework}, which involves certifying upper bounds of the probability that the system eventually enters unsafe sets from an initial state while adhering to state-constrained sets. Please refer to Remark \ref{safety_reach} in Subsection \ref{rv_1}. Moreover, under certain assumptions, we will formulate necessary and sufficient  barrier-like conditions for certifying $\mathbb{P}_{\pi}(RA_{\bm{x}_0}) \geq \epsilon_2$.

\section{Safety Verification}
\label{sec:rav}
This section introduces necessary and sufficient  barrier-like conditions for certifying lower bounds in safety verification and will detail their construction process. The construction involves constructing and relaxing a Bellman equation, one of whose solutions characterizes the exact safety probability $\mathbb{P}_{\pi}(S_{\bm{x}})$ for $\bm{x}\in \mathbb{R}^n$. The Bellman equation is derived from a value function.

We begin by introducing the value function $V(\cdot):
\mathbb{R}^n \rightarrow \mathbb{R}$, which characterizes the exact safety probability $\mathbb{P}_{\pi}(S_{\bm{x}})$ for each state $\bm{x}\in \mathbb{R}^n$,
\begin{equation}
    \label{value_safe}
    \begin{split}
    V(\bm{x}):=&\mathbb{E}_{\pi}\big[ g(\bm{x})\big],
    \end{split}
\end{equation}
where \[g(\bm{x})=1_{\mathbb{R}^n\setminus \mathcal{X}}(\bm{\phi}_{\pi}^{\bm{x}}(0))+\sum_{i\in \mathbb{N}_{\geq 1}}\prod_{j=0}^{i-1} 1_{\mathcal{X}}(\bm{\phi}_{\pi}^{\bm{x}}(j)) 1_{\mathbb{R}^n\setminus \mathcal{X}}(\bm{\phi}_{\pi}^{\bm{x}}(i)).\]

\begin{lemma}
\label{safe_equal}
The value function $V(\bm{x})$ in \eqref{value_safe} is equal to one minus the safety probability $\mathbb{P}(S_{\bm{x}})$, i.e., 
\[V(\bm{x})=1-\mathbb{P}_{\pi}(S_{\bm{x}})\]
for $\bm{x}\in \mathbb{R}^n$.
\end{lemma}
\begin{proof}  
By definition, $\mathbb{E}_{\pi}[1_{\mathbb{R}^n\setminus \mathcal{X}}(\bm{\phi}_{\pi}^{\bm{x}}(0))]=1_{\mathbb{R}^n\setminus \mathcal{X}}(\bm{x})$ holds. Furthermore, since
\begin{equation*}
    \begin{split}
        \mathbb{E}_{\pi}[\prod_{j=0}^{i-1} 1_{\mathcal{X}}(\bm{\phi}_{\pi}^{\bm{x}}(j)) 1_{\mathbb{R}^n\setminus \mathcal{X}}(\bm{\phi}_{\pi}^{\bm{x}}(i))]=\mathbb{P}_{\pi}(\wedge_{j=0}^{i-1}[\bm{\phi}_{\pi}^{\bm{x}}(j)\in \mathcal{X}] \wedge [\bm{\phi}_{\pi}^{\bm{x}}(i)\in \mathbb{R}^n\setminus \mathcal{X}])
    \end{split}
\end{equation*} is the probability that the system \eqref{system} starting from $\bm{x}$ will exit the safe set $\mathcal{X}$ at time $t=i$ while stay within $\mathcal{X}$ before $i$, where $i\in \mathbb{N}_{\geq 1}$, we have 
\begin{equation*}
\begin{split}
&\mathbb{E}_{\pi}[1_{\mathbb{R}^n\setminus \mathcal{X}}(\bm{\phi}_{\pi}^{\bm{x}}(0))]+\sum_{i\in \mathbb{N}_{\geq 1}}\mathbb{E}_{\pi}[\prod_{j=0}^{i-1} 1_{\mathcal{X}}(\bm{\phi}_{\pi}^{\bm{x}}(j)) 1_{\mathbb{R}^n\setminus \mathcal{X}}(\bm{\phi}_{\pi}^{\bm{x}}(i))]\\
&=\mathbb{P}_{\pi}(\bm{\phi}_{\pi}^{\bm{x}}(0)\in \mathbb{R}^n\setminus \mathcal{X})+\sum_{i\in \mathbb{N}_{\geq 1}} \mathbb{P}_{\pi}(\wedge_{j=0}^{i-1}[\bm{\phi}_{\pi}^{\bm{x}}(j)\in \mathcal{X}] \wedge [\bm{\phi}_{\pi}^{\bm{x}}(i)\in \mathbb{R}^n\setminus \mathcal{X}])\\
&=\mathbb{P}_{\pi}(\exists i\in \mathbb{N}. \bm{\phi}_{\pi}^{\bm{x}}(i)\in \mathbb{R}^n\setminus \mathcal{X}).
\end{split}
\end{equation*}
Thus, $\mathbb{E}_{\pi}[g(\bm{x})]=\mathbb{P}_{\pi}(\exists i\in \mathbb{N}. \bm{\phi}_{\pi}^{\bm{x}}(i)\in \mathbb{R}^n\setminus \mathcal{X})$. Consequently, $\mathbb{P}_{\pi}(S_{\bm{x}})=1-V(\bm{x})$. 
\end{proof}

According to Lemma \ref{safe_equal}, $V(\bm{x})$ falls within [0,1] for $\bm{x}\in \mathbb{R}^n$ and thus it is bounded over $\mathbb{R}^n$. We next will show that the value function \eqref{value_safe} can be reduced to a bounded solution to a Bellman equation (or, dynamic programming equation) via the dynamic programming principle. A value function characterizes the exact safety probability over finite-time horizons and its related dynamic programming equations can be found in \cite{abate2008probabilistic,laurenti2023unifying}.
\begin{proposition}
\label{pro_bellman1}
  The value function $V(\cdot):\mathbb{R}^n \rightarrow \mathbb{R}$ in \eqref{value_safe}
 satisfies the following Bellman equation
 \begin{equation}
 \label{bellman_safe}
     V(\bm{x})=1_{\mathbb{R}^n\setminus \mathcal{X}}(\bm{x})+ 1_{\mathcal{X}}(\bm{x}) \mathbb{E}_{\bm{\theta}}[V(\bm{f}(\bm{x},\bm{\theta}))]
 \end{equation}
for $\bm{x}\in \mathbb{R}^n$.
 \end{proposition}
 \begin{proof}
 Since $g(\bm{x})=1_{\mathbb{R}^n\setminus \mathcal{X}}(\bm{x})+1_{\mathcal{X}}(\bm{x})(1_{\mathbb{R}^n\setminus \mathcal{X}}(\bm{\phi}_{\pi}^{\bm{y}}(0))+\sum_{i\in \mathbb{N}_{\geq 1}}\prod_{j=0}^{i-1} 1_{\mathcal{X}}(\bm{\phi}_{\pi}^{\bm{y}}(j))1_{\mathbb{R}^n\setminus \mathcal{X}}(\bm{\phi}_{\pi}^{\bm{y}}(i)))$, we have
 \begin{equation*}
     \begin{split}
         V(\bm{x})=&1_{\mathbb{R}^n\setminus \mathcal{X}}(\bm{x})+1_{\mathcal{X}}(\bm{x})\mathbb{E}_{\pi}\left[
         1_{\mathbb{R}^n\setminus \mathcal{X}}(\bm{y})+\sum_{i\in \mathbb{N}_{\geq 1}}\prod_{j=0}^{i-1} 1_{\mathcal{X}}(\bm{\phi}_{\pi}^{\bm{y}}(j)) 1_{\mathbb{R}^n\setminus \mathcal{X}}(\bm{x}(i))\right]\\
         =&1_{\mathbb{R}^n\setminus \mathcal{X}}(\bm{x})+1_{\mathcal{X}}(\bm{x})\mathbb{E}_{\bm{\theta}}\left[1_{\mathbb{R}^n\setminus \mathcal{X}}(\bm{y})+\mathbb{E}_{\pi}[\sum_{i\in \mathbb{N}_{\geq 1}}\prod_{j=0}^{i-1} 1_{\mathcal{X}}(\bm{\phi}_{\pi}^{\bm{y}}(j)) 1_{\mathbb{R}^n\setminus \mathcal{X}}(\bm{x}(i))]\right]\\
        =&1_{\mathbb{R}^n\setminus \mathcal{X}}(\bm{x})+1_{\mathcal{X}}(\bm{x})\mathbb{E}_{\bm{\theta}}[V(\bm{y})]\\
         =&1_{\mathbb{R}^n\setminus \mathcal{X}}(\bm{x})+1_{\mathcal{X}}(\bm{x})\mathbb{E}_{\bm{\theta}}[V(\bm{f}(\bm{x},\bm{\theta}))],\\
     \end{split}
 \end{equation*}
 where  $\bm{y}=\bm{\phi}_{\pi}^{\bm{x}}(1)=\bm{f}(\bm{x},\bm{\theta})$.
 \end{proof}

It is observed that the Bellman equation \eqref{bellman_safe} may have multiple bounded solutions, since 
\[V'(\bm{x}):=V(\bm{x})+C\mathbb{E}_{\pi}[\prod_{j\in \mathbb{N}}1_{\mathcal{X}}(\bm{\phi}_{\pi}^{\bm{x}}(j))]\] also satisfies the equation \eqref{bellman_safe}, where $C$ is a constant and  $\mathbb{E}_{\pi}[\prod_{j\in \mathbb{N}}1_{\mathcal{X}}(\bm{\phi}_{\pi}^{\bm{x}}(j))]$ equals the safety probability that the system \eqref{system} starting from $\bm{x}$ will stay within the set $\mathcal{X}$ for all time.  Specially, when $C=1$, $V'(\bm{x})=1$ for $\bm{x}\in \mathbb{R}^n$ satisfies the Bellman equation \eqref{bellman_safe}.

A necessary and sufficient  barrier-like condition for certifying lower bounds in the safety verification can be derived via relaxing the Bellman equation \eqref{bellman_safe}.

\begin{theorem}
\label{upper_condition}
 Let $\epsilon_1\in [0,1]$.  There exists a function $v(\bm{x}): \mathbb{R}^n \rightarrow \mathbb{R}$ satisfying the following barrier-like condition:
 \begin{equation}
 \label{constraints1_safe}
     \begin{cases}
         v(\bm{x}_0)\leq 1-\epsilon_1, \\
         v(\bm{x}) \geq \mathbb{E}_{\bm{\theta}}[v(\bm{f}(\bm{x},\bm{\theta}))], & \forall \bm{x}\in \mathcal{X},\\
         v(\bm{x})\geq 1, &\forall \bm{x}\in \mathbb{R}^n\setminus \mathcal{X}, \\
         v(\bm{x}) \geq 0, & \forall \bm{x}\in \mathbb{R}^n,
     \end{cases}
 \end{equation}
 if and only if 
 $\mathbb{P}_{\pi}(S_{\bm{x}_0})\geq \epsilon_1$.
\end{theorem}
\begin{proof}
1) We first prove the ``only if'' part. 

We first prove  via induction that for all $k\in \mathbb{N}$,
\[
\begin{split}
\zeta_k(\bm{x}) :&= \mathbb{E}_{\pi}\left[ \sum_{i=0}^{k} \prod_{j=0}^{i-1} 1_{\mathcal{X}}(\bm{\phi}^{\bm{x}}_{\pi}(j)) \cdot 1_{\mathbb{R}^{n}\setminus\mathcal{X}}(\bm{\phi}^{\bm{x}}_{\pi}(i)) \right] + \mathbb{E}_{\pi}\left[ \prod_{j=0}^{k} 1_{\mathcal{X}}(\bm{\phi}^{\bm{x}}_{\pi}(j)) \cdot v(\bm{\phi}^{\bm{x}}_{\pi}(k+1)) \right]\leq v(\bm{x}).
\end{split}
\]

\textbf{Base Case ($k=0$):}
\[
\begin{split}
\zeta_0(\bm{x}) &= \mathbb{E}_{\pi}\left[1_{\mathbb{R}^{n}\setminus\mathcal{X}}(\bm{\phi}^{\bm{x}}_{\pi}(0))\right] + \mathbb{E}_{\pi}\left[1_{\mathcal{X}}(\bm{\phi}^{\bm{x}}_{\pi}(0)) v(\bm{\phi}^{\bm{x}}_{\pi}(1))\right] \\
&= 1_{\mathbb{R}^{n}\setminus\mathcal{X}}(\bm{x}) + 1_{\mathcal{X}}(\bm{x}) \mathbb{E}_{\bm{\theta}}[v(\bm{f}(\bm{x},\bm{\theta}))] \leq v(\bm{x}),
\end{split}
\]
where the first equality follows from the convention that the empty product equals 1, and the inequality follows from condition \eqref{constraints1_safe}.

\textbf{Inductive Step:} Assume $v(\bm{x}) \geq \zeta_k(\bm{x})$ for some $k \geq 0$. Then:
\begin{align*}
\zeta_{k+1}(\bm{x}) 
&= \zeta_k(\bm{x}) 
- \mathbb{E}_{\pi} \Bigg[ 
\prod_{j=0}^{k} 1_{\mathcal{X}}\big(\bm{\phi}^{\bm{x}}_{\pi}(j)\big) \,
v\big(\bm{\phi}^{\bm{x}}_{\pi}(k+1)\big) 
\Bigg] \\
&+ \mathbb{E}_{\pi} \Bigg[ 
\prod_{j=0}^{k} 1_{\mathcal{X}}\big(\bm{\phi}^{\bm{x}}_{\pi}(j)\big) \Bigg( 
1_{\mathbb{R}^{n} \setminus \mathcal{X}}\big(\bm{\phi}^{\bm{x}}_{\pi}(k+1)\big) + 1_{\mathcal{X}}\big(\bm{\phi}^{\bm{x}}_{\pi}(k+1)\big) \,
\mathbb{E}_{\bm{\theta}} \bigg[
v\big(\bm{\phi}^{\bm{x}}_{\pi}(k+2)\big)
\bigg] \Bigg) 
\Bigg].
\end{align*}

Using condition \eqref{constraints1_safe} at state $\bm{\phi}^{\bm{x}}_{\pi}(k+1)$:
\[
\begin{split}
v(\bm{\phi}^{\bm{x}}_{\pi}(k+1)) \geq &1_{\mathbb{R}^{n}\setminus\mathcal{X}}(\bm{\phi}^{\bm{x}}_{\pi}(k+1)) + 1_{\mathcal{X}}(\bm{\phi}^{\bm{x}}_{\pi}(k+1)) \mathbb{E}_{\bm{\theta}}[v(\bm{\phi}^{\bm{x}}_{\pi}(k+2))],
\end{split}
\]
we have $\zeta_{k+1}(\bm{x}) \leq \zeta_k(\bm{x}) \leq v(\bm{x})$.

By induction, $v(\bm{x}) \geq \zeta_k(\bm{x})$ for all $k \in \mathbb{N}$.  Since $\zeta_k(\bm{x})\geq 0$ for all $k \in \mathbb{N}$, $\lim_{k\to\infty} \zeta_k(\bm{x})$ exists. Taking $k \to \infty$, we have
\[
\begin{split}
\lim_{k\to\infty} \zeta_k(\bm{x}) &= \mathbb{E}_{\pi}\left[ \sum_{i=0}^{\infty} \prod_{j=0}^{i-1} 1_{\mathcal{X}}(\bm{\phi}^{\bm{x}}_{\pi}(j)) \cdot 1_{\mathbb{R}^{n}\setminus\mathcal{X}}(\bm{\phi}^{\bm{x}}_{\pi}(i)) \right] + \lim_{k\to\infty} \mathbb{E}_{\pi}\left[ \prod_{j=0}^{k} 1_{\mathcal{X}}(\bm{\phi}^{\bm{x}}_{\pi}(j)) v(\bm{\phi}^{\bm{x}}_{\pi}(k+1)) \right]\\
&\geq \mathbb{E}_{\pi}\left[ \sum_{i=0}^{\infty} \prod_{j=0}^{i-1} 1_{\mathcal{X}}(\bm{\phi}^{\bm{x}}_{\pi}(j)) \cdot 1_{\mathbb{R}^{n}\setminus\mathcal{X}}(\bm{\phi}^{\bm{x}}_{\pi}(i)) \right]\\
&=V(\bm{x}).
\end{split}
\]
 Thus, $v(\bm{x}) \geq V(\bm{x})$.

2) We will prove the ``if'' part.

If  $\mathbb{P}_{\pi}(S_{\bm{x}_0})\geq \epsilon_1$, we have $V(\bm{x}_0)\leq 1-\epsilon_1$ from Lemma \ref{safe_equal}, where $V(\cdot):\mathbb{R}^n \rightarrow \mathbb{R}$ is the value function in \eqref{value_safe}. Moreover, according to Proposition \ref{pro_bellman1}, $V(\bm{x})$ satisfies 
\begin{equation*}
\begin{cases}
V(\bm{x}) = \mathbb{E}_{\bm{\theta}}[V(\bm{f}(\bm{x},\bm{\theta}))], &\forall \bm{x}\in \mathcal{X},\\
V(\bm{x})=1, &\forall \bm{x}\in \mathbb{R}^n\setminus \mathcal{X}.
\end{cases}
\end{equation*}
Also, since $V(\bm{x}) \geq 0$ for $\bm{x}\in \mathbb{R}^n$, $V(\bm{x})$ satisfies \eqref{constraints1_safe}. 
\end{proof}

\begin{remark}
\label{remark1}
In this study, we consider the safety verification with respect to a fixed initial state $\bm{x}_0\in \mathcal{X}$. However, if we use an initial set $\mathcal{X}_0$, which is a set of initial states, the barrier-like condition \eqref{constraints1_safe}, with $v(\bm{x})\leq 1-\epsilon_1, \forall \bm{x}\in \mathcal{X}_0$ replacing $v(\bm{x}_0)\leq 1-\epsilon_1$, is also a necessary and sufficient  one for justifying $\mathbb{P}_{\pi}(S_{\bm{x}})\geq \epsilon_1, \forall \bm{x}\in \mathcal{X}_0$, since $\mathbb{P}_{\pi}(S_{\bm{x}})\geq \epsilon_1, \forall \bm{x}\in \mathcal{X}_0$ is equivalent to $V(\bm{x})\leq 1-\epsilon_1, \forall \bm{x}\in \mathcal{X}_0$, where $V(\cdot): \mathbb{R}^n\rightarrow \mathbb{R}$ is the value function \eqref{value_safe}.

In addition, the set $\mathbb{R}^n$ in condition \eqref{constraints1_safe} can be substituted with a set $\Omega$, which encompasses the reachable set of system \eqref{system} starting from the safe set $\mathcal{X}$ within a single step, i.e., 
\begin{equation}
    \label{omega}
    \Omega \supseteq \{\bm{x}_1 \mid \bm{x}_1 = \bm{f}(\bm{x}, \bm{\theta}), \forall \bm{x} \in \mathcal{X}, \bm{\theta} \in \Theta\} \cup \mathcal{X}.
\end{equation} The resulting condition also serves as both a necessary and sufficient  criterion for certifying lower bounds of safety probabilities. It  is the one (9) in Proposition 3 in \cite{yu2023safe}, which was derived using an auxiliary switched system and Ville's Inequality \cite{ville1939etude}. In \cite{yu2023safe}, only the sufficiency of the condition for safety verification was demonstrated. In addition, this condition serves as a typical example of the condition (3) with $\alpha=1$ and $\beta=0$ in Theorem 1 of \cite{xue2024finite}, which investigates finite-time safety verification.  It is important to note that while Proposition 2 in \cite{santoyo2021barrier} also establishes a sufficient barrier-like condition for certifying upper bounds on the safety probability of avoiding unsafe sets when $\tilde{\alpha}=1$ and $\tilde{\beta}=0$, the safety probability pertains to a stopped process that stops evolving upon exiting the set $\mathcal{X}$. For interested readers, please refer to Proposition 2 in \cite{santoyo2021barrier}. However, as discussed in Section \ref{intro} and in Remark \ref{safety_reach}, which is introduced later, the safety probability should be interpreted as the reach-avoid probability defined in Definition \ref{reach-avoid}.
$\blacksquare$
\end{remark}

\section{Reach-avoid Verification}
\label{sec:lb}
This section presents necessary and sufficient  barrier-like conditions for the reach-avoid verification in Definition \ref{reach-avoid}. Two cases are discussed in this section. The first case assumes that, for every state in $\mathcal{X}\setminus \mathcal{X}_r$, the system \eqref{system} will either leave the safe set $\mathcal{X}$ or enter the target set $\mathcal{X}_r$ in finite time almost surely. The second case considers the assumption that the specified lower bound $\epsilon_2$ is strictly less than the exact reach-avoid probability $\mathbb{P}_{\pi}(RA_{\bm{x}_0})$, i.e., $\epsilon_2<\mathbb{P}_{\pi}(RA_{\bm{x}_0})$. These two cases are detailed in Subsection \ref{rv_1} and \ref{rv_2}, respectively.

\subsection{Reach-avoid Verification I}
\label{rv_1}
This subsection formulates a necessary and sufficient barrier-like condition for reach-avoid verification, under the assumption that, for every state in $\mathcal{X} \setminus \mathcal{X}_r$, the system \eqref{system} will almost surely either enter the target set $\mathcal{X}_r$ or exit the safe set $\mathcal{X}$ in finite time. Similar to the one in Section \ref{sec:rav}, this condition is also constructed by relaxing a Bellman equation.  The Bellman equation is derived from a value function.

We begin by introducing the value function $V(\cdot):\mathbb{R}^n \rightarrow \mathbb{R}$, which characterizes the exact reach-avoid probability $\mathbb{P}_{\pi}(RA_{\bm{x}})$ for $\bm{x}\in \mathbb{R}^n$. We define the value function as follows:
\begin{equation}
    \label{value}
    \begin{split}
    V(\bm{x}):=&\mathbb{E}_{\pi}\big[ g(\bm{x})\big],
    \end{split}
\end{equation}
where \[g(\bm{x})=1_{\mathcal{X}_r}(\bm{\phi}_{\pi}^{\bm{x}}(0))+\sum_{i\in \mathbb{N}_{\geq 1}}\prod_{j=0}^{i-1} 1_{\mathcal{X}\setminus \mathcal{X}_r}(\bm{\phi}_{\pi}^{\bm{x}}(j)) 1_{\mathcal{X}_r}(\bm{\phi}_{\pi}^{\bm{x}}(i)).\]

\begin{lemma}
\label{equal}
The value function $V(\bm{x})$ in \eqref{value} is equal to the reach-avoid probability $\mathbb{P}_{\pi}(RA_{\bm{x}})$, i.e., $V(\bm{x})=\mathbb{P}_{\pi}(RA_{\bm{x}})$
for $\bm{x}\in \mathbb{R}^n$.
\end{lemma}
\begin{proof}  
By definition, $\mathbb{E}_{\pi}[1_{\mathcal{X}_r}(\bm{\phi}_{\pi}^{\bm{x}}(0))]=1_{\mathcal{X}_r}(\bm{x})$. In addition, since 
$\mathbb{E}_{\pi}[\prod_{j=0}^{i-1} 1_{\mathcal{X}\setminus \mathcal{X}_r}(\bm{\phi}_{\pi}^{\bm{x}}(j)) 1_{\mathcal{X}_r}(\bm{\phi}_{\pi}^{\bm{x}}(i))]=\mathbb{P}_{\pi}(\wedge_{j=0}^{i-1}[\bm{\phi}_{\pi}^{\bm{x}}(j)\in \mathcal{X}\setminus \mathcal{X}_r] \wedge [\bm{\phi}_{\pi}^{\bm{x}}(i)\in \mathcal{X}_r])$ is the probability that the system \eqref{system} starting from $\bm{x}$ will enter the set $\mathcal{X}_r$ at time $t=i$ while staying within $\mathcal{X}\setminus \mathcal{X}_r$ before $i$, where $i\geq 1$. Thus, we have 
\begin{equation*}
\begin{split}
&\mathbb{E}_{\pi}[1_{\mathcal{X}_r}(\bm{\phi}_{\pi}^{\bm{x}}(0))]+\sum_{i\in \mathbb{N}_{\geq 1}}\mathbb{E}_{\pi}[\prod_{j=0}^{i-1} 1_{\mathcal{X}\setminus \mathcal{X}_r}(\bm{\phi}_{\pi}^{\bm{x}}(j)) 1_{\mathbb{R}^n\setminus \mathcal{X}_r}(\bm{\phi}_{\pi}^{\bm{x}}(i))]\\
&=\mathbb{P}_{\pi}(\bm{\phi}_{\pi}^{\bm{x}}(0)\in \mathcal{X}_r)+\sum_{i\in \mathbb{N}_{\geq 1}} \mathbb{P}_{\pi}(\wedge_{j=0}^{i-1}[\bm{\phi}_{\pi}^{\bm{x}}(j)\in \mathcal{X}\setminus \mathcal{X}_r] \wedge [\bm{\phi}_{\pi}^{\bm{x}}(i)\in \mathcal{X}_r])\\
&=\mathbb{P}_{\pi}(RA_{\bm{x}}).
\end{split}
\end{equation*}
Consequently, $\mathbb{P}_{\pi}(RA_{\bm{x}})=V(\bm{x}).$
\end{proof}

We next will show that the value function \eqref{value} can be reduced to a solution to a Bellman equation via the dynamic programming principle. A controlled version of the Bellman equation can be found in \cite{summers2010verification}.
\begin{proposition}
\label{pro_bellman}
    The value function $V(\cdot):\mathbb{R}^n \rightarrow \mathbb{R}$ in \eqref{value}
 satisfies the following Bellman equation
 \begin{equation}
 \label{bellman}
     V(\bm{x})=1_{\mathcal{X}_r}(\bm{x})+ 1_{\mathcal{X}\setminus \mathcal{X}_r}(\bm{x}) \mathbb{E}_{\bm{\theta}}[V(\bm{f}(\bm{x},\bm{\theta}))]
 \end{equation}
 for $\bm{x}\in \mathbb{R}^n$.
 \end{proposition}
 \begin{proof}
  Since $g(\bm{x})=1_{\mathcal{X}_r}(\bm{x})+1_{\mathcal{X}\setminus \mathcal{X}_r}(\bm{x})(1_{\mathcal{X}_r}(\bm{y})+\sum_{i\in \mathbb{N}_{\geq 1}}\prod_{j=0}^{i-1} 1_{\mathcal{X}}(\bm{\phi}_{\pi}^{\bm{y}}(j)) 1_{\mathcal{X}_r}(\bm{\phi}_{\pi}^{\bm{y}}(i)))$, we have 
 \begin{equation*}
     \begin{split}
         V(\bm{x})=&1_{\mathcal{X}_r}(\bm{x})+1_{\mathcal{X}\setminus \mathcal{X}_r}(\bm{x})\mathbb{E}_{\pi}[1_{\mathcal{X}_r}(\bm{y})+\sum_{i\in \mathbb{N}_{\geq 1}}\prod_{j=0}^{i-1} 1_{\mathcal{X}\setminus \mathcal{X}_r}(\bm{\phi}_{\pi}^{\bm{y}}(j)) 1_{\mathcal{X}_r}(\bm{\phi}_{\pi}^{\bm{y}}(i))]\\
         =&1_{\mathcal{X}_r}(\bm{x})+1_{\mathcal{X}\setminus \mathcal{X}_r}(\bm{x})\mathbb{E}_{\bm{\theta}}[1_{\mathcal{X}_r}(\bm{y})+\mathbb{E}_{\pi}\left[\sum_{i\in \mathbb{N}_{\geq 1}}\prod_{j=0}^{i-1} 1_{\mathcal{X}\setminus \mathcal{X}_r}(\bm{\phi}_{\pi}^{\bm{y}}(j)) 1_{\mathcal{X}_r}(\bm{\phi}_{\pi}^{\bm{y}}(i))]\right]\\
         =&1_{\mathcal{X}_r}(\bm{x})+1_{\mathcal{X}\setminus \mathcal{X}_r}(\bm{x})\mathbb{E}_{\bm{\theta}}[V(\bm{f}(\bm{x},\bm{\theta}))]\\
     \end{split}
 \end{equation*}
 where $\bm{y}=\bm{\phi}_{\pi}^{\bm{x}}(1)=\bm{f}(\bm{x},\bm{\theta})$.
 \end{proof}

\begin{remark}
\label{safety_reach}
Similar to the condition \eqref{constraints1_safe} in Theorem \ref{upper_condition},  we can also construct a necessary and sufficient  condition for the safety verification scenario in \cite{prajna2007framework}, which is certifying upper bounds of the probability that the system eventually enters unsafe sets
from an initial state while adhering to state-constrained sets, by relaxing the Bellman equation \eqref{bellman}. It is shown in Proposition \ref{upper_condition_ra}. The proof is shown in Appendix. In this proposition, $\mathcal{X}_r$ is a set of unsafe states and $\mathcal{X}$ is a state-constrained set. This condition is also a typical instance of condition (9) with $\alpha=1$ and $\beta=0$ in Theorem 3 in \cite{xue2024finite}, which provides upper bounds of the reach-avoid probability in the finite-time reach-avoid verification.  
\begin{proposition}
\label{upper_condition_ra}
 Let $\epsilon'_1\in [0,1]$. There exists a function $v(\bm{x}): \mathbb{R}^n \rightarrow \mathbb{R}$ satisfying the barrier-like condition:
 \begin{equation}
 \label{constraints1}
     \begin{cases}
         v(\bm{x}_0)\leq \epsilon'_1, \\
         v(\bm{x}) \geq \mathbb{E}_{\bm{\theta}}[v(\bm{f}(\bm{x},\bm{\theta}))], & \forall \bm{x}\in \mathcal{X}\setminus \mathcal{X}_r,\\
         v(\bm{x})\geq 1, &\forall \bm{x}\in \mathcal{X}_r, \\
         v(\bm{x}) \geq 0, & \forall \bm{x}\in \mathbb{R}^n\setminus \mathcal{X},
     \end{cases}
 \end{equation}
 if and only if 
 $\mathbb{P}_{\pi}(S'_{\bm{x}_0})\leq \epsilon'_1$, where $S'_{\bm{x}_0}=RA_{\bm{x}_0}=\{\pi \mid \exists k\in \mathbb{N}. \bm{\phi}_{\pi}^{\bm{x}_0}(k)\in \mathcal{X}_r \wedge \forall i\in \mathbb{N}_{\leq k}.  \bm{\phi}_{\pi}^{\bm{x}_0}(i)\in \mathcal{X}\}$.
\end{proposition}

As discussed in Remark \ref{remark1}, we can also revise condition \eqref{constraints1} to establish a necessary and sufficient criterion for ensuring that 
 $\mathbb{P}_{\pi}(S'_{\bm{x}})\leq \epsilon'_1, \forall \bm{x}\in \mathcal{X}_0$, where $\mathcal{X}_0$ is a set of initial states.

In addition, as discussed in Remark \ref{remark1}, a sufficient barrier-like condition is formulated in Proposition 2 with parameters $\tilde{\alpha}=1$ and $\tilde{\beta}=0$ in \cite{santoyo2021barrier}. This condition can also be used for certifying upper bounds for the probability $\mathbb{P}_{\pi}(S'_{\bm{x}_0})$. The primary distinction between this condition and the one presented in \eqref{constraints1} is that the barrier function $ B(\bm{x})$ in Proposition 2 in \cite{santoyo2021barrier} does not require the condition $ B(\bm{x}) \geq 0 $ for $ \bm{x} \in \mathbb{R}^n \setminus \mathcal{X} $.
$\blacksquare$
\end{remark}

However, it is generally not feasible to formulate necessary and sufficient  conditions for certifying lower bounds in the reach-avoid verification by relaxing the Bellman equation \eqref{bellman}. The underlying reason is that the bounded solutions to the Bellman equation \eqref{bellman} are typically non-unique. Nevertheless, under certain assumptions, we can ensure uniqueness of these solutions, thereby enabling the derivation of such conditions. 
\begin{assumption}
\label{assump1}
For every initial state $\bm{x}\in \mathcal{X}\setminus \mathcal{X}_r$, the system \eqref{system} exits the set $\mathcal{X}\setminus \mathcal{X}_r$ in finite time almost surely; that is, 
$\mathbb{P}_{\pi}(\forall k\in \mathbb{N}. \bm{\phi}_{\pi}^{\bm{x}}(k)\in \mathcal{X}\setminus \mathcal{X}_r)=0, \forall \bm{x}\in \mathcal{X}\setminus \mathcal{X}_r$.
\end{assumption}

There are systems satisfying Assumption \ref{assump1}. For instance,  consider the stochastic system $\bm{x}(k+1)=\bm{g}(\bm{x}(k))+\bm{\theta}(k)$, where $\bm{\theta}(k)$ is a i.i.d. Gaussian disturbance. Let the initial state $\bm{x}(0)$ lie within a bounded set $\mathcal{X}\setminus \mathcal{X}_r$. Since the additive noise has unbounded support, there is a non-zero probability that the trajectory will eventually exit any bounded set. In fact, with probability one, the trajectory will leave $\mathcal{X}\setminus \mathcal{X}_r$ in finite time. Therefore, the condition $\mathbb{P}_{\pi}(\forall k\in \mathbb{N}. \bm{\phi}_{\pi}^{\bm{x}}(k)\in \mathcal{X}\setminus \mathcal{X}_r)=0, \forall \bm{x}\in \mathcal{X}\setminus \mathcal{X}_r$ holds,
satisfying Assumption \ref{assump1}.

\begin{proposition}
\label{uni_solution}
Under Assumption \ref{assump1}, the Bellman equation \eqref{bellman} has a unique bounded solution over $\mathbb{R}^n$, which is the value function \eqref{value}.
\end{proposition}
\begin{proof}  
As shown in Proposition \ref{pro_bellman}, the value function \eqref{value} satisfies the Bellman equation \eqref{bellman}. 

In the following, we just show that if a bounded function $v(\bm{x}): \mathbb{R}^n \rightarrow \mathbb{R}$ satisfies the Bellman equation \eqref{bellman}, $v(\bm{x})=V(\bm{x})$ holds for $\bm{x}\in \mathbb{R}^n$. 

We first show that $
v(\bm{x}) = V(\bm{x}) + 1_{\mathcal{X}\setminus \mathcal{X}_r}(\bm{x}) \lim_{k \to \infty} h_k(\bm{x})$ for $\bm{x}\in \mathbb{R}^n$, where 
\[h_k(\bm{x}) :=\mathbb{E}_{\pi} \left[ \prod_{j=1}^k 1_{\mathcal{X}\setminus \mathcal{X}_r}(\bm{\phi}_\pi^{\bm{x}}(j)) \cdot v(\bm{\phi}_\pi^{\bm{x}}(k+1)) \right]
\]
for $k\in \mathbb{N}$. We note that when $k=0$, the product is taken over an empty index set, and by convention, the empty product equals 1. Thus, 
$h_0(\bm{x})=\mathbb{E}_{\pi}[v(\bm{\phi}_\pi^{\bm{x}}(1))]=\mathbb{E}_{\bm{\theta}}[v(\bm{f}(\bm{x},\bm{\theta}))]$.

For this sake, we prove by induction that for all \(k \in \mathbb{N}\),
\[
\begin{split}
\zeta_k(\bm{x}) :&= \underbrace{\mathbb{E}_{\pi}\left[ \sum_{i=0}^k \prod_{j=0}^{i-1} 1_{\mathcal{X}\setminus \mathcal{X}_r}(\bm{\phi}_\pi^{\bm{x}}(j)) \cdot 1_{\mathcal{X}_r}(\bm{\phi}_\pi^{\bm{x}}(i)) \right]}_{V_k(\bm{x})}+ \mathbb{E}_{\pi}\left[ \prod_{j=0}^k 1_{\mathcal{X}\setminus \mathcal{X}_r}(\bm{\phi}_\pi^{\bm{x}}(j)) \cdot v(\bm{\phi}_\pi^{\bm{x}}(k+1)) \right]= v(\bm{x}).
\end{split}
\]

\textbf{Base case} \(k=0\):
\[
\begin{split}
\zeta_0(\bm{x}) 
&= \mathbb{E}_{\pi}\left[ 1_{\mathcal{X}_r}(\bm{\phi}_\pi^{\bm{x}}(0)) \right]
+ \mathbb{E}_{\pi}\left[ 1_{\mathcal{X}\setminus \mathcal{X}_r}(\bm{\phi}_\pi^{\bm{x}}(0)) \cdot v(\bm{\phi}_\pi^{\bm{x}}(1)) \right] \\
&= 1_{\mathcal{X}_r}(\bm{x}) + 1_{\mathcal{X}\setminus \mathcal{X}_r}(\bm{x}) \cdot \mathbb{E}_{\bm{\theta}} \left[ v(\bm{f}(\bm{x}, \bm{\theta})) \right] = v(\bm{x}),
\end{split}
\]
where the last equality follows from the fixed-point condition on \(v\).

\textbf{Inductive step:} Assume the statement holds for some \(k \geq 0\), i.e., \(\zeta_k(\bm{x}) = v(\bm{x})\). Then,
\begin{align*}
\zeta_{k+1}(\bm{x}) 
= &\mathbb{E}_{\pi} \left[ \sum_{i=0}^{k+1} \prod_{j=0}^{i-1} 1_{\mathcal{X}\setminus \mathcal{X}_r}(\bm{\phi}_\pi^{\bm{x}}(j)) \cdot 1_{\mathcal{X}_r}(\bm{\phi}_\pi^{\bm{x}}(i)) \right]
+ \mathbb{E}_{\pi} \left[ \prod_{j=0}^{k+1} 1_{\mathcal{X}\setminus \mathcal{X}_r}(\bm{\phi}_\pi^{\bm{x}}(j)) \cdot v(\bm{\phi}_\pi^{\bm{x}}(k+2)) \right] \\
=& \underbrace{\mathbb{E}_{\pi} \left[ \sum_{i=0}^{k} \prod_{j=0}^{i-1} 1_{\mathcal{X}\setminus \mathcal{X}_r}(\bm{\phi}_\pi^{\bm{x}}(j)) \cdot 1_{\mathcal{X}_r}(\bm{\phi}_\pi^{\bm{x}}(i)) \right]}_{= V_k(\bm{x})}  + \mathbb{E}_{\pi} \left[ \prod_{j=0}^k 1_{\mathcal{X}\setminus \mathcal{X}_r}(\bm{\phi}_\pi^{\bm{x}}(j)) \cdot 1_{\mathcal{X}_r}(\bm{\phi}_\pi^{\bm{x}}(k+1)) \right] \\
&\quad\quad\quad\quad\quad\quad\quad\quad\quad\quad\quad\quad\quad\quad\quad\quad + \mathbb{E}_{\pi} \left[ \prod_{j=0}^{k+1} 1_{\mathcal{X}\setminus \mathcal{X}_r}(\bm{\phi}_\pi^{\bm{x}}(j)) \cdot v(\bm{\phi}_\pi^{\bm{x}}(k+2)) \right] \\
=& \zeta_k(\bm{x}) - \mathbb{E}_{\pi} \left[ \prod_{j=0}^k 1_{\mathcal{X}\setminus \mathcal{X}_r}(\bm{\phi}_\pi^{\bm{x}}(j)) \cdot v(\bm{\phi}_\pi^{\bm{x}}(k+1)) \right] \\
&\quad + \mathbb{E}_{\pi} \Bigg[ \prod_{j=0}^k 1_{\mathcal{X}\setminus \mathcal{X}_r}(\bm{\phi}_\pi^{\bm{x}}(j)) \cdot \Big( 1_{\mathcal{X}_r}(\bm{\phi}_\pi^{\bm{x}}(k+1))+ 1_{\mathcal{X}\setminus \mathcal{X}_r}(\bm{\phi}_\pi^{\bm{x}}(k+1)) \cdot \mathbb{E}_{\bm{\theta}} \left[ v(\bm{\phi}_\pi^{\bm{x}}(k+2)) \right] \Big)  \Bigg] \\
=& \zeta_k(\bm{x}) -1_{\mathcal{X}\setminus \mathcal{X}_r}(\bm{x})h_k(\bm{x}) + 1_{\mathcal{X}\setminus \mathcal{X}_r}(\bm{x})h_k(\bm{x}) \\
= &\zeta_k(\bm{x}) = v(\bm{x}),
\end{align*}
where $v(\bm{\phi}_\pi^{\bm{x}}(k+1))=1_{\mathcal{X}_r}(\bm{\phi}_\pi^{\bm{x}}(k+1)) + 1_{\mathcal{X}\setminus \mathcal{X}_r}(\bm{\phi}_\pi^{\bm{x}}(k+1)) \cdot \mathbb{E}_{\bm{\theta}} \left[ v(\bm{\phi}_\pi^{\bm{x}}(k+2)) \right]$, which can be obtained via the Bellman equation \eqref{bellman}.

By induction, \(\zeta_k(\bm{x}) = v(\bm{x})\) for all \(k\geq 0\).

Finally, taking the limit as \(k \to \infty\), we have
\[
\begin{split}
v(\bm{x})& = \lim_{k \to \infty} \zeta_k(\bm{x}) \\
&= \underbrace{\mathbb{E}_{\pi} \left[ \sum_{i=0}^\infty \prod_{j=0}^{i-1} 1_{\mathcal{X}\setminus \mathcal{X}_r}(\bm{\phi}_\pi^{\bm{x}}(j)) \cdot 1_{\mathcal{X}_r}(\bm{\phi}_\pi^{\bm{x}}(i)) \right]}_{= V(\bm{x})} + 1_{\mathcal{X}\setminus \mathcal{X}_r}(\bm{x})\lim_{k \to \infty} h_k(\bm{x}),
\end{split}
\]
which establishes
\begin{equation}
\label{9}
v(\bm{x}) = V(\bm{x}) + 1_{\mathcal{X}\setminus \mathcal{X}_r}(\bm{x}) \lim_{k \to \infty} h_k(\bm{x}).
\end{equation}

In the following, we show $\lim_{k\rightarrow \infty}  h_k(\bm{x})=0$ for $\bm{x}\in \mathcal{X}\setminus \mathcal{X}_r$. 

1).   For all \(\bm{x} \in \mathcal{X} \setminus \mathcal{X}_r\), the system \eqref{system} exits \(\mathcal{X} \setminus \mathcal{X}_r\) in finite time almost surely, i.e.,
   \[
   \mathbb{P}_{\pi} \left( \forall k \in \mathbb{N}.  \bm{\phi}^{\bm{x}}_{\pi}(k) \in \mathcal{X} \setminus \mathcal{X}_r \right) = 0.
   \]
   This implies that trajectories starting in \(\mathcal{X} \setminus \mathcal{X}_r\) will almost surely either
   \begin{enumerate}
       \item enter the target set \(\mathcal{X}_r\), or  
       \item leave the safe set \(\mathcal{X}\) (i.e., enter \(\mathbb{R}^n \setminus \mathcal{X}\))  
   \end{enumerate}
   within finite time.

2). The function \(v(\cdot)\) is bounded over \(\mathbb{R}^n\). Thus, there exists a constant \(M > 0\) such that 
   \[
   |v(\bm{y})| \leq M, \quad \forall \bm{y} \in \mathbb{R}^n.
   \]

3). Define the event:
   \[
   A_k(\bm{x}) :=\left\{ \bm{\phi}^{\bm{x}}_{\pi}(j) \in \mathcal{X} \setminus \mathcal{X}_r \text{ for all } j = 1, 2, \dots, k \right\}.
   \]
   This represents the set of disturbance signals \(\pi\) where the trajectory remains in \(\mathcal{X} \setminus \mathcal{X}_r\) from time \(1\) to \(k\).

4). Assumption \ref{assump1} implies that the event \(\bigcap_{i=1}^{\infty} A_i(\bm{x})\) has probability zero, i.e.,
   \[
   \mathbb{P}_{\pi} \left( \bigcap_{k=1}^{\infty} A_k(\bm{x}) \right) = \mathbb{P}_{\pi} \left( \forall k \geq 1.  \bm{\phi}^{\bm{x}}_{\pi}(k) \in \mathcal{X} \setminus \mathcal{X}_r \right) = 0.
   \]
   Since \(A_{k+1}(\bm{x}) \subseteq A_k(\bm{x})\) (the sequence is nested), continuity of probability measures gives
   \[
   \lim_{k \to \infty} \mathbb{P}_{\pi} \left( A_k(\bm{x}) \right) = \mathbb{P}_{\pi} \left( \bigcap_{k=1}^{\infty} A_k(\bm{x}) \right) = 0.
   \]

5). The term \(\prod_{j=1}^{k} 1_{\mathcal{X} \setminus \mathcal{X}_r} (\bm{\phi}^{\bm{x}}_{\pi}(j))\) is the indicator of \(A_k(\bm{x})\). Using the boundedness of \(v\) (assume $|v| \leq M$ over $\mathcal{X}\setminus \mathcal{X}_r$), we have 
   \[
   \begin{split}
   |h_k(\bm{x})| &\leq \mathbb{E}_{\pi} \left[ \prod_{j=1}^{k} 1_{\mathcal{X} \setminus \mathcal{X}_r} (\bm{\phi}_\pi^{\bm{x}}(j)) \cdot |v(\bm{\phi}_\pi^{\bm{x}}(k+1)| \right] \leq M \cdot \mathbb{E}_{\pi} \left[ \prod_{j=1}^{k} 1_{\mathcal{X} \setminus \mathcal{X}_r} (\bm{\phi}_\pi^{\bm{x}}(j)) \right] = M \cdot \mathbb{P}_{\pi} \left( A_k(\bm{x}) \right).
   \end{split}
   \]
   As \(k \to \infty\), we have
   \[
   \lim_{k \to \infty} |h_k(\bm{x})| \leq M \cdot \lim_{k \to \infty} \mathbb{P}_{\pi} \left( A_k(\bm{x}) \right) = 0.
   \]
   Hence, \(\lim_{k \to \infty} h_k(\bm{x}) = 0\) for $\bm{x}\in \mathcal{X}\setminus \mathcal{X}_r$. Consequently, $h_k(\bm{x})=0$ for $\bm{x}\in \mathcal{X}\setminus \mathcal{X}_r$.

Finally, when \(\bm{x} \in \mathcal{X}_r\) or \(\bm{x} \in \mathbb{R}^n \setminus \mathcal{X}\), the term \(1_{\mathcal{X} \setminus \mathcal{X}_r}(\bm{x})\) in \eqref{9} is zero, thus $\lim_{k\rightarrow \infty} 1_{\mathcal{X}\setminus \mathcal{X}_r}(\bm{x}) h_k(\bm{x})=0$.  

Consequently, $v(\bm{x})=V(\bm{x})$ over $\mathbb{R}^n$. 
\end{proof}

Under Assumption \ref{assump1}, we can establish necessary and sufficient conditions for certifying lower bounds in reach-avoid verification by relaxing the Bellman equation \eqref{bellman}.

\begin{theorem}
\label{lower_condition}
 Let $\epsilon_2\in [0,1]$.  Under Assumption \ref{assump1}, there exists a function $v(\bm{x}): \mathbb{R}^n \rightarrow \mathbb{R}$, which is bounded in $\mathcal{X}$ and satisfies the following condition:
  \begin{equation}
 \label{constraints11}
     \begin{cases}
         v(\bm{x}_0)\geq \epsilon_2, \\
         v(\bm{x}) \leq \mathbb{E}_{\bm{\theta}}[v(\bm{f}(\bm{x},\bm{\theta}))], & \forall \bm{x}\in \mathcal{X}\setminus \mathcal{X}_r,\\
         v(\bm{x})\leq 1, &\forall \bm{x}\in \mathcal{X}_r, \\
         v(\bm{x})\leq 0,  & \forall \bm{x}\in \mathbb{R}^n\setminus \mathcal{X},
     \end{cases}
 \end{equation}
 if and only if 
 $\mathbb{P}_{\pi}(RA_{\bm{x}_0})\geq \epsilon_2$.
\end{theorem}
\begin{proof}
1) We first prove the ``only if'' part. 

Since $v(\bm{x})$ satisfies \eqref{constraints11}, by following the inductive argument used in the proof of Proposition \ref{uni_solution}-where we showed \[v(\bm{x}) = V(\bm{x}) + 1_{\mathcal{X}\setminus \mathcal{X}_r}(\bm{x}) \lim_{k \to \infty} h_k(\bm{x})\]
-but replacing the equality ``='' with ``$\leq$'', we obtain that 
\begin{equation*}
\begin{split}
v(\bm{x})\leq V(\bm{x})+1_{\mathcal{X}\setminus \mathcal{X}_r}(\bm{x})\lim_{k\rightarrow \infty} h_k(\bm{x})
\end{split}
\end{equation*}
for $\bm{x}\in \mathbb{R}^n$, where $V(\cdot): \mathbb{R}^n \rightarrow \mathbb{R}$ is the value function defined in \eqref{value}.
Since $v(\bm{\phi}_{\pi}^{\bm{x}}(k+1))\leq 0$ when $\bm{\phi}_{\pi}^{\bm{x}}(k+1)\in \mathbb{R}^n\setminus \mathcal{X}$, we have 
\begin{equation*}
\begin{split}
    h_k(\bm{x})&=\mathbb{E}_{\pi}[\prod_{j=1}^{k}1_{\mathcal{X}\setminus \mathcal{X}_r}(\bm{\phi}_{\pi}^{\bm{x}}(j))v(\bm{\phi}_{\pi}^{\bm{x}}(k+1))]\leq \mathbb{E}_{\pi}[\prod_{j=1}^{k}1_{\mathcal{X}\setminus \mathcal{X}_r}(\bm{\phi}_{\pi}^{\bm{x}}(j))w_{k+1}(\bm{x})],
    \end{split}
\end{equation*}
where $w_{k+1}(\bm{x})=1_{\mathcal{X}}(\bm{\phi}_{\pi}^{\bm{x}}(k+1))v(\bm{\phi}_{\pi}^{\bm{x}}(k+1))$. Also, since 
$v(\cdot): \mathbb{R}^n \rightarrow \mathbb{R}$ is bounded over $\mathcal{X}$ and 
$\mathbb{P}_{\pi}(\forall k\in \mathbb{N}. \bm{\phi}_{\pi}^{\bm{x}}(k)\in \mathcal{X}\setminus \mathcal{X}_r)=0$ for $\bm{x}\in \mathcal{X}\setminus \mathcal{X}_r$, we conclude 
$\lim_{k\rightarrow \infty} 1_{\mathcal{X}\setminus \mathcal{X}_r}(\bm{x}) h_k(\bm{x})=0$ for $\bm{x}\in \mathbb{R}^n$. Consequently, $v(\bm{x})\leq V(\bm{x})$ for $\bm{x}\in \mathbb{R}^n$.

Thus, $\mathbb{P}_{\pi}(RA_{\bm{x}_0})=V(\bm{x}_0)\geq v(\bm{x}_0) \geq \epsilon_2$.

2) We will prove the ``if'' part.

If  $\mathbb{P}_{\pi}(RA_{\bm{x}_0})\geq \epsilon_2$, we have $V(\bm{x}_0)\geq \epsilon_2$ according to Lemma \ref{equal}, where $V(\cdot):\mathbb{R}^n \rightarrow \mathbb{R}$ is the value function in \eqref{value}. Moreover, according to Proposition \ref{pro_bellman}, $V(\bm{x})$ satisfies 
\begin{equation*}
\begin{cases}
V(\bm{x}) = \mathbb{E}_{\bm{\theta}}[V(\bm{f}(\bm{x},\bm{\theta}))], &\forall \bm{x}\in \mathcal{X}\setminus \mathcal{X}_r,\\
V(\bm{x})=1, &\forall \bm{x}\in \mathcal{X}_r,\\
V(\bm{x})=0, &\forall \bm{x}\in \mathbb{R}^n\setminus \mathcal{X}.
\end{cases}
\end{equation*}
Consequently, $V(\bm{x})$ satisfies \eqref{constraints11}. 
\end{proof}

\begin{remark}
There is an important distinction between Proposition \ref{upper_condition_ra} and Theorem \ref{lower_condition} that we now clarify explicitly here:

1) Proposition \ref{upper_condition_ra} provides a condition for verifying upper bounds on reach-avoid probabilities and does not require Assumption \ref{assump1}. This makes it broadly applicable, particularly in settings where the system may remain within $\mathcal{X}\setminus \mathcal{X}_r$ indefinitely with nonzero probability.

2) Theorem \ref{lower_condition}, on the other hand, provides necessary and sufficient conditions for verifying lower bounds on reach-avoid probabilities. However, it relies on Assumption \ref{assump1}, which ensures that the probability of the system \eqref{system} staying in 
$\mathcal{X}\setminus \mathcal{X}_r$ for all time is zero. This assumption is essential to guarantee that, if there exists a function satisfying condition (10), then the specified threshold $\epsilon_2$ is indeed a valid lower bound for the reach-avoid probability $\mathbb{P}_{\pi}(RA_{\bm{x}_0})$. Without Assumption \ref{assump1}, we cannot use condition \eqref{constraints11} to justify lower bounds in the reach-avoid verification, since we cannot guarantee $\lim_{i\rightarrow \infty} h_i(\bm{x})=0$ for $\bm{x}\in \mathcal{X}\setminus \mathcal{X}_r$. 
\end{remark}

\begin{remark}
\label{remark4}
As discussed in Remark \ref{remark1}, we can also revise condition \eqref{constraints11} to establish a necessary and sufficient criterion for ensuring that 
 $\mathbb{P}_{\pi}(RA_{\bm{x}})\geq \epsilon_2, \forall \bm{x}\in \mathcal{X}_0$, where $\mathcal{X}_0$ is a set of initial states.
\end{remark}

\subsection{Reach-avoid Verification II}
\label{rv_2}
The subsection will formulate a necessary and sufficient  barrier-like condition for the reach-avoid verification without Assumption \ref{assump1}. Instead, another assumption that $\epsilon_2$ is strictly smaller than the exact reach-avoid probability $\mathbb{P}_{\pi}(RA_{\bm{x}_0})$ is imposed. Similar to the one in Subsection \ref{rv_1}, this condition is constructed by relaxing a Bellman equation, which is derived from a discounted value function.

Let's start with the discounted value function $\tilde{V}_{\gamma}(\cdot):\mathbb{R}^n \rightarrow \mathbb{R}$,
\begin{equation}
    \label{value_gamma}
    \tilde{V}_{\gamma}(\bm{x}):=\mathbb{E}_{\pi}[\tilde{g}_{\gamma}(\bm{x})],
\end{equation}
where 
\[\tilde{g}_{\gamma}(\bm{x})=1_{\mathcal{X}_r}(\bm{\phi}_{\pi}^{\bm{x}}(0))+\sum_{i\in \mathbb{N}_{\geq 1}}\gamma^{i} \prod_{j=0}^{i-1} 1_{\mathcal{X}\setminus \mathcal{X}_r}(\bm{\phi}_{\pi}^{\bm{x}}(j)) 1_{\mathcal{X}_r}(\bm{\phi}_{\pi}^{\bm{x}}(i))\] 
and $\gamma \in [0,1]$ is a user-defined value. 

The value $\tilde{V}_{\gamma}(\bm{x})$ in \eqref{value_gamma} is a lower bound of the exact reach-avoid probability $\mathbb{P}_{\pi}(RA_{\bm{x}})$ for $\bm{x}\in \mathbb{R}^n$. Moreover, when $\gamma$ approaches $1$,  $\tilde{V}_{\gamma}(\bm{x})$ will approach $\mathbb{P}_{\pi}(RA_{\bm{x}})$ for $\bm{x}\in \mathbb{R}^n$. 

\begin{lemma}
\label{lower_convergence}
For $\bm{x}\in \mathbb{R}^n$,
\[\tilde{V}_{\gamma}(\bm{x})\leq \mathbb{P}_{\pi}(RA_{\bm{x}})\] and
\[\lim_{\gamma\rightarrow 1^-} \tilde{V}_{\gamma}(\bm{x})=\mathbb{P}_{\pi}(RA_{\bm{x}}),\]
where  $\tilde{V}(\cdot):\mathbb{R}^n \rightarrow \mathbb{R}$ is the value function in \eqref{value_gamma}.
\end{lemma}
\begin{proof}
  The conclusion $\tilde{V}_{\gamma}(\bm{x})\leq \mathbb{P}_{\pi}(RA_{\bm{x}})$ can be justified according to $\gamma\in [0,1]$  and Lemma \ref{equal}.

In the following, we just show $\lim_{\gamma\rightarrow 1^-} \tilde{V}_{\gamma}(\bm{x})=\mathbb{P}_{\pi}(RA_{\bm{x}})$. 

1) We first show $\tilde{V}_{\gamma}(\bm{x})$ is uniformly convergent over $\gamma \in [0,1]$.  According to Lemma \ref{equal}, $\mathbb{P}_{\pi}(RA_{\bm{x}})=V(\bm{x})$, where $V(\cdot): \mathbb{R}^n \rightarrow \mathbb{R}$ is the value function in \eqref{value}. Thus, for every $\epsilon>0$, there exists $N\in \mathbb{N}$ such that 
\[\sum_{k=m+1}^M \mathbb{E}_{\pi} [\prod_{j=0}^{k-1} 1_{\mathcal{X}\setminus \mathcal{X}_r}(\bm{\phi}_{\pi}^{\bm{x}}(j)) 1_{\mathcal{X}_r}(\bm{\phi}_{\pi}^{\bm{x}}(k))]<\epsilon, \forall M>m>N,\]
where $M,m\in \mathbb{N}$.
Since 
\[
\begin{split}
    &\sum_{k=m+1}^M \mathbb{E}_{\pi} [\gamma^{k} \prod_{j=0}^{k-1} 1_{\mathcal{X}\setminus \mathcal{X}_r}(\bm{\phi}_{\pi}^{\bm{x}}(j)) 1_{\mathcal{X}_r}(\bm{\phi}_{\pi}^{\bm{x}}(k))]
    \leq \sum_{k=m+1}^M \mathbb{E}_{\pi} [\prod_{j=0}^{k-1} 1_{\mathcal{X}\setminus \mathcal{X}_r}(\bm{\phi}_{\pi}^{\bm{x}}(j)) 1_{\mathcal{X}_r}(\bm{\phi}_{\pi}^{\bm{x}}(k))]
    \end{split}
    \]
    holds for $\gamma \in [0,1]$, we have $\tilde{V}_{\gamma}(\bm{x})$ is uniformly convergent over $\gamma \in [0,1]$. 

    In addition, $\mathbb{E}_{\pi}[\gamma^{i} \prod_{j=0}^{i-1} 1_{\mathcal{X}\setminus \mathcal{X}_r}(\bm{\phi}_{\pi}^{\bm{x}}(j)) 1_{\mathcal{X}_r}(\bm{\phi}_{\pi}^{\bm{x}}(i))]$ is continuous over $\gamma\in [0,1]$, where $i \in \mathbb{N}_{\geq 1}$. Therefore, according to Term-by-term Continuity Theorem, we obtain $\lim_{\gamma\rightarrow 1^-} \tilde{V}_{\gamma}(\bm{x})=\mathbb{E}_{\pi}[1_{\mathcal{X}_r}(\bm{\phi}_{\pi}^{\bm{x}}(0))+\sum_{i\in \mathbb{N}_{\geq 1}}\lim_{\gamma\rightarrow 1^-}\gamma^{i} \prod_{j=0}^{i-1} 1_{\mathcal{X}\setminus \mathcal{X}_r}(\bm{\phi}_{\pi}^{\bm{x}}(j)) 1_{\mathcal{X}_r}(\bm{\phi}_{\pi}^{\bm{x}}(i))]=V(\bm{x})=\mathbb{P}_{\pi}(RA_{\bm{x}})$.
  \end{proof}

  \begin{proposition}
  \label{pro_gamma}
When $\gamma \in [0,1)$, the value function \eqref{value_gamma} $\tilde{V}_{\gamma}(\cdot): \mathbb{R}^n \rightarrow \mathbb{R}$ in \eqref{value_gamma} satisfies the following Bellman equation:
 \begin{equation}
     \label{bellman_gamma}
     \tilde{V}_{\gamma}(\bm{x})=1_{\mathcal{X}_r}(\bm{x})+\gamma 1_{\mathcal{X}\setminus \mathcal{X}_r}(\bm{x})\mathbb{E}_{\bm{\theta}}[\tilde{V}_{\gamma}(\bm{f}(\bm{x},\bm{\theta}))]
 \end{equation}
 for $\bm{x}\in \mathbb{R}^n$. Moreover, the Bellman equation \eqref{bellman_gamma} possess a unique bounded solution over $\mathbb{R}^n$. 
  \end{proposition}
\begin{proof}  
The conclusion that the value function \eqref{value_gamma} satisfies the Bellman equation \eqref{bellman_gamma} can be justified by following the proof of Proposition \ref{pro_bellman}.

In the following, we just show that if a bounded function $v(\bm{x}): \mathbb{R}^n \rightarrow \mathbb{R}$ satisfies the Bellman equation \eqref{bellman_gamma}, $v(\bm{x})=\tilde{V}_{\gamma}(\bm{x})$ holds for $\bm{x}\in \mathbb{R}^n$. 

We first show that \[
v(\bm{x}) = \tilde{V}_{\gamma}(\bm{x}) + 1_{\mathcal{X}\setminus \mathcal{X}_r}(\bm{x}) \lim_{k \to \infty} h_k(\bm{x})
\]
for $\bm{x}\in \mathbb{R}^n$, where 
\[
h_k(\bm{x}) := \gamma^{k+1}\mathbb{E}_{\pi} \left[ \prod_{j=1}^k 1_{\mathcal{X}\setminus \mathcal{X}_r}(\bm{\phi}_\pi^{\bm{x}}(j)) \cdot v(\bm{\phi}_\pi^{\bm{x}}(k+1)) \right].
\]
We note that when $k=0$, the product is taken over an empty index set, and by convention, the empty product equals 1. Therefore, 
$h_0(\bm{x})=\gamma \mathbb{E}_{\pi}[v(\bm{\phi}_\pi^{\bm{x}}(1))]=\gamma \mathbb{E}_{\bm{\theta}}[v(\bm{f}(\bm{x},\bm{\theta}))].$

For this sake, we prove by induction that for all \(k \in \mathbb{N}\),
\[
\begin{split}
\zeta_k(\bm{x}) :&= \underbrace{\mathbb{E}_{\pi} \left[ \sum_{i=0}^{k} \gamma^i \prod_{j=0}^{i-1} 1_{\mathcal{X}\setminus \mathcal{X}_r}(\bm{\phi}_\pi^{\bm{x}}(j)) \cdot 1_{\mathcal{X}_r}(\bm{\phi}_\pi^{\bm{x}}(i)) \right]}_{= \tilde{V}_{k,\gamma}(\bm{x})} 
+ \gamma^{k+1}\mathbb{E}_{\pi}\left[ \prod_{j=0}^k 1_{\mathcal{X}\setminus \mathcal{X}_r}(\bm{\phi}_\pi^{\bm{x}}(j)) \cdot v(\bm{\phi}_\pi^{\bm{x}}(k+1)) \right]\\
&= v(\bm{x}).
\end{split}
\]

\textbf{Base case} \(k=0\):
\[
\begin{split}
\zeta_0(\bm{x}) 
&= \mathbb{E}_{\pi}\left[ 1_{\mathcal{X}_r}(\bm{\phi}_\pi^{\bm{x}}(0)) \right]
+\gamma  \mathbb{E}_{\pi}\left[ 1_{\mathcal{X}\setminus \mathcal{X}_r}(\bm{\phi}_\pi^{\bm{x}}(0)) \cdot v(\bm{\phi}_\pi^{\bm{x}}(1)) \right] \\
&= 1_{\mathcal{X}_r}(\bm{x}) + \gamma 1_{\mathcal{X}\setminus \mathcal{X}_r}(\bm{x}) \cdot \mathbb{E}_{\bm{\theta}} \left[ v(\bm{f}(\bm{x}, \bm{\theta})) \right] \\
&= v(\bm{x}),
\end{split}
\]
where the last equality follows from the Bellman equation \eqref{bellman_gamma}.

\textbf{Inductive step:} Assume the statement holds for some \(k \geq 0\), i.e., \(\zeta_k(\bm{x}) = v(\bm{x})\). Then,
\begin{align*}
&\zeta_{k+1}(\bm{x}) \\
&= \mathbb{E}_{\pi} \left[ \sum_{i=0}^{k+1} \gamma^i \prod_{j=0}^{i-1} 1_{\mathcal{X}\setminus \mathcal{X}_r}(\bm{\phi}_\pi^{\bm{x}}(j)) \cdot 1_{\mathcal{X}_r}(\bm{\phi}_\pi^{\bm{x}}(i)) \right]+ \gamma^{k+2} \mathbb{E}_{\pi} \left[ \prod_{j=0}^{k+1} 1_{\mathcal{X}\setminus \mathcal{X}_r}(\bm{\phi}_\pi^{\bm{x}}(j)) \cdot v(\bm{\phi}_\pi^{\bm{x}}(k+2)) \right] \\
&= \underbrace{\mathbb{E}_{\pi} \left[ \sum_{i=0}^{k} \gamma^i \prod_{j=0}^{i-1} 1_{\mathcal{X}\setminus \mathcal{X}_r}(\bm{\phi}_\pi^{\bm{x}}(j)) \cdot 1_{\mathcal{X}_r}(\bm{\phi}_\pi^{\bm{x}}(i)) \right]}_{= \tilde{V}_{k,\gamma}(\bm{x})} + \mathbb{E}_{\pi} \left[ \gamma^{k+1} \prod_{j=0}^k 1_{\mathcal{X}\setminus \mathcal{X}_r}(\bm{\phi}_\pi^{\bm{x}}(j)) \cdot 1_{\mathcal{X}_r}(\bm{\phi}_\pi^{\bm{x}}(k+1)) \right] \\
&\quad+ \gamma^{k+2}\mathbb{E}_{\pi} \left[ \prod_{j=0}^{k+1} 1_{\mathcal{X}\setminus \mathcal{X}_r}(\bm{\phi}_\pi^{\bm{x}}(j)) \cdot v(\bm{\phi}_\pi^{\bm{x}}(k+2)) \right] \\
&=\zeta_k(\bm{x}) - \gamma^{k+1}\mathbb{E}_{\pi} \left[ \prod_{j=0}^k 1_{\mathcal{X}\setminus \mathcal{X}_r}(\bm{\phi}_\pi^{\bm{x}}(j)) \cdot v(\bm{\phi}_\pi^{\bm{x}}(k+1)) \right] \\
&\quad\quad+ \gamma^{k+1}\mathbb{E}_{\pi} \Bigg[ \prod_{j=0}^k 1_{\mathcal{X}\setminus \mathcal{X}_r}(\bm{\phi}_\pi^{\bm{x}}(j)) \cdot \Big( 1_{\mathcal{X}_r}(\bm{\phi}_\pi^{\bm{x}}(k+1)) + \gamma 1_{\mathcal{X}\setminus \mathcal{X}_r}(\bm{\phi}_\pi^{\bm{x}}(k+1)) \cdot \mathbb{E}_{\bm{\theta}} \left[ v(\bm{\phi}_\pi^{\bm{x}}(k+2)) \right] \Big)  \Bigg] \\
&=\zeta_k(\bm{x}) -1_{\mathcal{X}\setminus \mathcal{X}_r}(\bm{x})h_k(\bm{x}) + 1_{\mathcal{X}\setminus \mathcal{X}_r}(\bm{x})h_k(\bm{x}) \\
&=\zeta_k(\bm{x}) = v(\bm{x}).
\end{align*}

By induction, \(\zeta_k(\bm{x}) = v(\bm{x})\) for all \(k\).

Finally, taking the limit as \(k \to \infty\), we have
\[
\begin{split}
v(\bm{x}) = \lim_{k \to \infty} \zeta_k(\bm{x}) = \underbrace{\mathbb{E}_{\pi} \left[ \sum_{i=0}^\infty \gamma^{i}\prod_{j=0}^{i-1} 1_{\mathcal{X}\setminus \mathcal{X}_r}(\bm{\phi}_\pi^{\bm{x}}(j)) \cdot 1_{\mathcal{X}_r}(\bm{\phi}_\pi^{\bm{x}}(i)) \right]}_{= \tilde{V}_{\gamma}(\bm{x})} + 1_{\mathcal{X}\setminus \mathcal{X}_r}(\bm{x})\lim_{k \to \infty} h_k(\bm{x}),
\end{split}
\]
which establishes
\[
v(\bm{x}) = \tilde{V}_{\gamma}(\bm{x}) + 1_{\mathcal{X}\setminus \mathcal{X}_r}(\bm{x}) \lim_{k \to \infty} h_k(\bm{x}).
\]

Since $v(\cdot):\mathbb{R}^n \rightarrow \mathbb{R}$ is bounded over 
$\mathbb{R}^n$, we have $\lim_{k\rightarrow \infty} h_k(\bm{x})=0$ for $\bm{x}\in \mathbb{R}^n$ and consequently, $v(\bm{x})=\tilde{V}_{\gamma}(\bm{x})$ over $\mathbb{R}^n$. 
\end{proof}

We can construct a necessary and sufficient  barrier-like condition for the reach-avoid verification in Definition \ref{reach-avoid} by relaxing the Bellman equation \eqref{bellman_gamma}, under the assumption that the reach-avoid probability $\mathbb{P}_{\pi}(RA_{\bm{x}_0})$ is strictly larger than the threshold $\epsilon_2$. This condition is the stochastic version of the one in Corollary 1 in \cite{zhao2022inner}.
\begin{assumption}
\label{strict_small}
   The reach-avoid probability $\mathbb{P}_{\pi}(RA_{\bm{x}_0})$ is strictly larger than the threshold $\epsilon_2$, i.e., \[\mathbb{P}_{\pi}(RA_{\bm{x}_0})>\epsilon_2,\] where $\epsilon_2\in [0,1)$.
\end{assumption}

Assumption \ref{strict_small} is not overly restrictive and does not generally compromise the practical utility of the proposed method. In practice, since $\mathbb{P}_{\pi}(RA_{\bm{x}_0})$ is unknown, it is rare for the threshold $\epsilon_2$ set by engineers to exactly match $\mathbb{P}_{\pi}(RA_{\bm{x}_0})$. Therefore, either $\epsilon_2$ tends to be larger or smaller than $\mathbb{P}_{\pi}(RA_{\bm{x}_0})$, with both cases occurring frequently. When $\epsilon_2$ exceeds $\mathbb{P}_{\pi}(RA_{\bm{x}_0})$, certification is not possible, as the claim becomes infeasible. In contrast, the case where $\epsilon_2$ is smaller than $\mathbb{P}_{\pi}(RA_{\bm{x}_0})$, which corresponds to Assumption \ref{strict_small}, is the focus of this work.

\begin{theorem}
\label{thm3}
Let $\epsilon_2\in [0,1)$. If there exist a constant $\gamma \in (0,1)$ and a function $v(\bm{x}):\mathbb{R}^n \rightarrow \mathbb{R}$, which is bounded over $\mathcal{X}$ and satisfies the following condition: 
\begin{equation}
\label{constraint_gamma0}
    \begin{cases}
        v(\bm{x}_0)\geq \epsilon_2,\\
         v(\bm{x}) \leq \gamma \mathbb{E}_{\bm{\theta}}[v(\bm{f}(\bm{x},\bm{\theta}))], & \forall \bm{x}\in \mathcal{X}\setminus \mathcal{X}_r,\\
         v(\bm{x})\leq 1, &\forall \bm{x}\in \mathcal{X}_r, \\
         v(\bm{x})\leq 0, & \forall \bm{x}\in \mathbb{R}^n\setminus \mathcal{X},
    \end{cases}
\end{equation}
then, $\mathbb{P}_{\pi}(RA_{\bm{x}_0})\geq \epsilon_2$.  Moreover, under Assumption~\ref{strict_small}, there indeed exist such a constant 
$\gamma \in (0,1)$ and a function $v(\bm{x}):\mathbb{R}^n \rightarrow \mathbb{R}$, bounded over $\mathcal{X}$, that satisfy \eqref{constraint_gamma0}.
\end{theorem}
\begin{proof}
1) Since $v(\bm{x})$ satisfies \eqref{constraint_gamma0}, by following the inductive argument used in the proof of Proposition \ref{uni_solution}-where we showed \[v(\bm{x}) = \tilde{V}_{\gamma}(\bm{x}) + 1_{\mathcal{X}\setminus \mathcal{X}_r}(\bm{x}) \lim_{k \to \infty} h_k(\bm{x})\]
-but replacing the equality ``='' with ``$\leq$'', we obtain that 
\begin{equation*}
\begin{split}
v(\bm{x})\leq \tilde{V}_{\gamma}(\bm{x})+1_{\mathcal{X}\setminus \mathcal{X}_r}(\bm{x})\lim_{k\rightarrow \infty} h_k(\bm{x})
\end{split}
\end{equation*}
for $\bm{x}\in \mathbb{R}^n$, where $\tilde{V}_{\gamma}(\cdot): \mathbb{R}^n \rightarrow \mathbb{R}$ is the value function defined in \eqref{value_gamma}. Since $v(\bm{\phi}_{\pi}^{\bm{x}}(k+1))\leq 0$ when $\bm{\phi}_{\pi}^{\bm{x}}(k+1)\in \mathbb{R}^n\setminus \mathcal{X}$, we have
\begin{equation*}
\begin{split}
    h_k(\bm{x})&=\gamma^{k+1}\mathbb{E}_{\pi}[\prod_{j=1}^{k}1_{\mathcal{X}\setminus \mathcal{X}_r}(\bm{\phi}_{\pi}^{\bm{x}}(j))v(\bm{\phi}_{\pi}^{\bm{x}}(k+1))]\leq \gamma^{k+1}\mathbb{E}_{\pi}[\prod_{j=1}^{k}1_{\mathcal{X}\setminus \mathcal{X}_r}(\bm{\phi}_{\pi}^{\bm{x}}(j))w_{k+1}(\bm{x})],
    \end{split}
\end{equation*}
where $w_{k+1}(\bm{x})=1_{\mathcal{X}}(\bm{\phi}_{\pi}^{\bm{x}}(k+1))v(\bm{\phi}_{\pi}^{\bm{x}}(k+1))$. Also, since 
$v(\cdot): \mathbb{R}^n \rightarrow \mathbb{R}$ is bounded over $\mathcal{X}$ and 
$\lim_{k\rightarrow \infty} \gamma^{k+1}=0$, we conclude 
\[\lim_{k\rightarrow \infty} h_k(\bm{x})=0\] for $\bm{x}\in \mathbb{R}^n$. Consequently, $v(\bm{x})\leq \tilde{V}_{\gamma}(\bm{x})$ for $\bm{x}\in \mathbb{R}^n$.

Thus, $\mathbb{P}_{\pi}(RA_{\bm{x}_0})\geq \tilde{V}_{\gamma}(\bm{x}_0)\geq v(\bm{x}_0) \geq \epsilon_2$ according to Lemma \ref{lower_convergence}.

2) According to Lemma \ref{lower_convergence}, $\lim_{\gamma\rightarrow 1^-} \tilde{V}_{\gamma}(\bm{x}_0)=\mathbb{P}_{\pi}(RA_{\bm{x}_0})$ holds. Since $\mathbb{P}_{\pi}(RA_{\bm{x}_0})>\epsilon_2$, there exists $\gamma_0$ such that $\tilde{V}_{\gamma_0}(\bm{x}_0)\geq \epsilon_2$ according to Lemma \ref{lower_convergence}. Moreover, according to Proposition \ref{pro_gamma}, $\tilde{V}_{\gamma_0}(\bm{x})$ satisfies 
\begin{equation*}
    \begin{cases}
        \tilde{V}_{\gamma_0} = \gamma_0 \mathbb{E}_{\bm{\theta}}[\tilde{V}_{\gamma_0}(\bm{f}(\bm{x},\bm{\theta}))], &\forall \bm{x}\in \mathcal{X}\setminus \mathcal{X}_r,\\
        \tilde{V}_{\gamma_0}(\bm{x})= 1, &\forall \bm{x}\in \mathcal{X}_r,\\
        \tilde{V}_{\gamma_0}(\bm{x})= 0, &\forall \bm{x}\in \mathbb{R}^n\setminus \mathcal{X}.
    \end{cases}
\end{equation*} Consequently, $\tilde{V}_{\gamma_0}(\bm{x})$ satisfies \eqref{constraint_gamma0}. 
\end{proof}

\begin{remark}
\label{remark5}
    If we consider an initial set $\mathcal{X}_0\subseteq \mathcal{X}\setminus \mathcal{X}_r$, which includes  infinitely many initial states, rather than a fixed initial state $\bm{x}_0\in \mathcal{X}\setminus \mathcal{X}_r$, we cannot guarantee that there exist a constant $\gamma\in (0,1)$ and a function $v(\bm{x}):\mathbb{R}^n \rightarrow \mathbb{R}$, which is bounded over $\mathcal{X}$ and satisfies the condition \eqref{constraint_gamma0} with $v(\bm{x})\geq \epsilon_2, \forall \bm{x} \in \mathcal{X}_0$ replacing $v(\bm{x}_0)\geq \epsilon_2$, such that $\mathbb{P}_{\pi}(RA_{\bm{x}})\geq \epsilon_2, \forall \bm{x}\in \mathcal{X}_0$. This is because we cannot guarantee that $\lim_{\gamma\rightarrow 1^-} \tilde{V}_{\gamma}(\bm{x})=\mathbb{P}_{\pi}(RA_{\bm{x}})$ holds uniformly over $\mathcal{X}_0$. 
    
    In addition, condition \eqref{constraint_gamma0} is a typical instance of condition (13) with $\alpha>1$ and $\beta=0$ in Theorem 5 in \cite{xue2024finite}, which offers lower bounds of the reach-avoid probability in the context of finite-time reach-avoid verification.  $\blacksquare$
\end{remark}

\begin{remark}
We note here that we can also construct a necessary and sufficient  condition to certify upper bounds of the safety probability $\mathbb{P}_{\pi}(\forall k\in \mathbb{N}. \bm{\phi}^{\bm{x}_0}_{\pi}(k)\in \mathcal{X})$ such that the system \eqref{system} starting from the initial state $\bm{x}_0$ will stay within the safe set $\mathcal{X}$ for all time \cite{yu2023safe}, under the assumption that  $\mathbb{P}_{\pi}(\forall k\in \mathbb{N}. \bm{\phi}^{\bm{x}_0}_{\pi}(k)\in \mathcal{X})<1-\epsilon_1$. 
    Under the assumption that $\mathbb{P}_{\pi}(\forall k\in \mathbb{N}. \bm{\phi}^{\bm{x}_0}_{\pi}(k)\in \mathcal{X})<1-\epsilon_1$, there exist a constant $\gamma \in (0,1)$ and a function $v(\bm{x}):\mathbb{R}^n \rightarrow \mathbb{R}$, which is bounded over $\mathcal{X}$ and satisfies the following condition: 
\begin{equation}
\label{constraint_gamma}
    \begin{cases}
        v(\bm{x}_0)\geq \epsilon_1,\\
         v(\bm{x}) \leq \gamma \mathbb{E}_{\bm{\theta}}[v(\bm{f}(\bm{x},\bm{\theta}))], & \forall \bm{x}\in \mathcal{X},\\
         v(\bm{x})\leq 1, &\forall \bm{x}\in \mathbb{R}^n\setminus \mathcal{X}, 
    \end{cases}
\end{equation}
if and only if $\mathbb{P}_{\pi}(\forall k\in \mathbb{N}. \bm{\phi}^{\bm{x}_0}_{\pi}(k)\in \mathcal{X})\leq 1-\epsilon_1$ (or equivalently, $\mathbb{P}_{\pi}(\exists k\in \mathbb{N}. \bm{\phi}^{\bm{x}_0}_{\pi}\in \mathbb{R}^n\setminus \mathcal{X})\geq \epsilon_1$).  Condition \eqref{constraint_gamma} is also a typical instance of condition (6) with $\alpha>1$ and $\beta=0$ in Theorem 2 in \cite{xue2024finite}, which offers upper bounds of the safety probability in the finite-time safety verification. Such conditions for certifying upper bounds become particularly significant in scenarios where $\mathbb{R}^n \setminus \mathcal{X}$ represents the target set that the system aims to reach. In this context, the safety probability will be referred to as the liveness probability.
$\blacksquare$
\end{remark}

Based on the value function \eqref{value_gamma}, we are able to show the necessity of another sufficient barrier-like condition in \cite{xue2021reach} for the reach-avoid verification under Assumption \ref{strict_small}. The condition is presented below:
\begin{equation}
\label{w}
    \begin{cases}
        v(\bm{x}_0)\geq \epsilon_2,\\
        v(\bm{x})\leq  \mathbb{E}_{\bm{\theta}}[v(\bm{f}(\bm{x},\bm{\theta}))],& \forall \bm{x}\in \mathcal{X}\setminus \mathcal{X}_r,\\
        v(\bm{x})\leq \mathbb{E}_{\bm{\theta}}[w(\bm{f}(\bm{x},\bm{\theta}))]-w(\bm{x}), & \forall \bm{x}\in \mathcal{X}\setminus \mathcal{X}_r,\\
        v(\bm{x})\leq 1, & \forall \bm{x}\in \mathcal{X}_r,\\
        v(\bm{x})\leq 0, & \forall \bm{x}\in \Omega\setminus \mathcal{X},
    \end{cases}
\end{equation}
where $\Omega$ is a set in \eqref{omega}. If there exist a function $v(\cdot):\Omega\rightarrow \mathbb{R}$ and a bounded function $w(\cdot):\Omega\rightarrow \mathbb{R}$ satisfying \eqref{w}, $\mathbb{P}_{\pi}(RA_{\bm{x}_0})\geq \epsilon_2$ holds. This conclusion can be justified by following the proof of Corollary 2 in \cite{xue2021reach}. In the following, we just demonstrate its necessity. 
\begin{corollary}
If $\mathbb{P}_{\pi}(RA_{\bm{x}_0})>\epsilon_2$, then there exist a function $v(\cdot):\Omega\rightarrow \mathbb{R}$ and a bounded function $w(\cdot):\Omega\rightarrow \mathbb{R}$ satisfying \eqref{w}.
\end{corollary}
\begin{proof}
   According to Lemma \ref{lower_convergence}, there exists $\gamma_0\in (0,1)$ such that $\tilde{V}_{\gamma_0}(\bm{x}_0)\geq \epsilon_2$ holds. From \eqref{bellman_gamma}, we can obtain 
   \[
   \begin{cases}
    1\geq \tilde{V}_{\gamma_0}(\bm{x})\geq 0,& \forall \bm{x}\in \mathbb{R}^n,\\
   \tilde{V}_{\gamma_0}(\bm{x})=\gamma_0 \mathbb{E}_{\bm{\theta}}[\tilde{V}_{\gamma_0}(\bm{f}(\bm{x},\bm{\theta}))]\leq \mathbb{E}_{\bm{\theta}}[\tilde{V}_{\gamma_0}(\bm{f}(\bm{x},\bm{\theta}))], & \forall \bm{x}\in \mathcal{X}\setminus \mathcal{X}_r,\\
   \tilde{V}_{\gamma_0}(\bm{x})\leq 1, & \forall \bm{x}\in \mathcal{X}_r,\\
   \tilde{V}_{\gamma_0}(\bm{x})=0,  & \forall \bm{x}\in \Omega\setminus \mathcal{X}.
   \end{cases}
   \]
   Let $\gamma_1$ be a constant satisfying $\frac{\gamma_1}{1+\gamma_1}\geq \gamma_0$, and $w(\bm{x}):=\gamma_1 \tilde{V}_{\gamma_0}(\bm{x})$ for $\bm{x}\in \mathbb{R}^n$. Thus,  
   \begin{equation*}
   \begin{split}
   &\frac{\mathbb{E}_{\bm{\theta}}[w(\bm{f}(\bm{x},\bm{\theta}))]-w(\bm{x})-\tilde{V}_{\gamma_0}(\bm{x})}{1+\gamma_1}\\
   =&\frac{\gamma_1 \mathbb{E}_{\bm{\theta}}[\tilde{V}_{\gamma_0}(\bm{f}(\bm{x},\bm{\theta}))]-\gamma_1\tilde{V}_{\gamma_0}(\bm{x})-\tilde{V}_{\gamma_0}(\bm{x})}{1+\gamma_1}\\
   =&\frac{\gamma_1}{1+\gamma_1} \mathbb{E}_{\bm{\theta}}[\tilde{V}_{\gamma_0}(\bm{f}(\bm{x},\bm{\theta}))]-\tilde{V}_{\gamma_0}(\bm{x})\\
   \geq &\gamma_0 \mathbb{E}_{\bm{\theta}}[\tilde{V}_{\gamma_0}(\bm{f}(\bm{x},\bm{\theta}))]-\tilde{V}_{\gamma_0}(\bm{x})=0.
   \end{split}
   \end{equation*}
   Thus, the functions $\tilde{V}_{\gamma_0}(\bm{x})$ and $w(\bm{x}):=\gamma_1 \tilde{V}_{\gamma_0}(\bm{x})$ satisfy \eqref{w}. Consequently,  there exist a function $v(\cdot):\Omega\rightarrow \mathbb{R}$ and a bounded function $w(\cdot):\Omega\rightarrow \mathbb{R}$ satisfying \eqref{w}.  
   \end{proof}

\section{Examples}
\label{sec:ex}
In this section, we demonstrate the application of our theoretical developments through two examples. In both cases, the function $f(x,\theta)$ is a polynomial in the state variables $x$, and the safe set $\mathcal{X}$ as well as the target set $\mathcal{X}_r$ are semi-algebraic sets. We aim to search for polynomial barrier-like functions to solve the associated verification problem. To do this, we encode the constraints \eqref{constraints1_safe}, \eqref{constraints11}, and \eqref{constraint_gamma0} as semi-definite programs (SDPs) using the sum of squares (SOS) decomposition for multivariate polynomials. The resulting SDPs are then solved  the tool Mosek 10.1.21 \cite{aps2019mosek}. To ensure numerical stability during the solution of these SDPs, we impose a constraint on the coefficients of the unknown polynomials, specifically restricting them to the interval $[-100, 100]$. In the sequel, $\sum[x]$ denotes the set of sum-of-squares polynomials over variables $x$, i.e., 
$\sum[x]=\{p\in \mathbb{R}[x]\mid p=\sum_{i=1}^k q^2_i(x), q_i(x)\in \mathbb{R}[x],i=1,\ldots,k\}$, where $\mathbb{R}[x]$ denotes the ring of polynomials in variables $x$.

\begin{example}[Safety Verification]
\label{ex1}

Consider the one-dimensional discrete-time system:
\begin{equation}
x(l+1)=(-0.5+\theta(l))\,x(l),
\end{equation}
where $\theta(l)\in\Theta=[-1,1]$ is uniform, the safe set is $\mathcal{X}=\{x \mid h(x)\leq 0\}$ with $h(x)=x^2-1$, and the initial state is $x_0=-0.8$. We simulate $10^4$ trajectories over $10^4$ time steps. The estimated safety probability is $0.8286$. Figure~\ref{fig:ex1_traj} shows three example trajectories over 10 time steps.
 
\begin{figure}[htbp]
    \centering
    \includegraphics[width=0.45\textwidth]{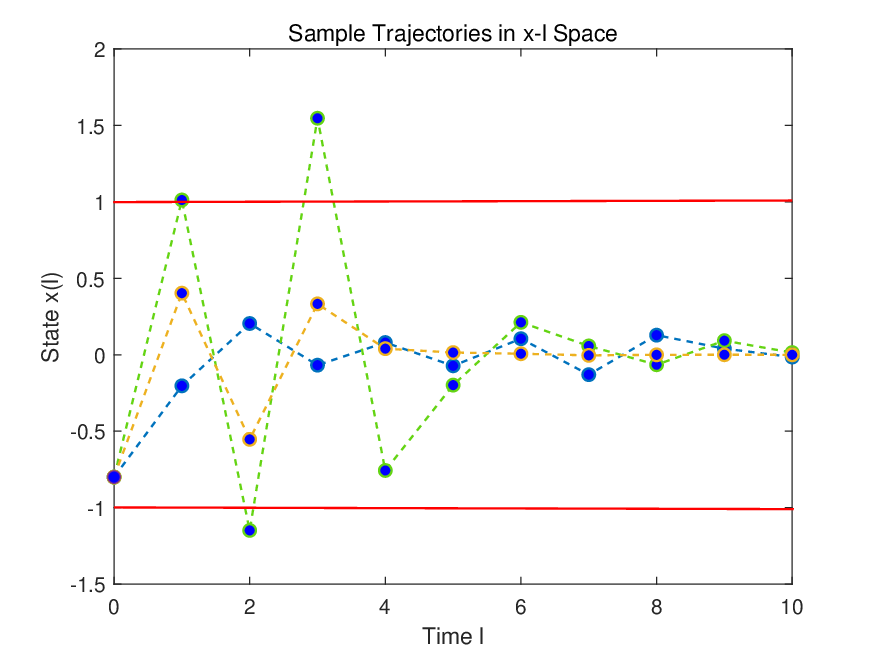} 
    \caption{Trajectories of system in Example~\ref{ex1} with $x_0=-0.8$.The region enclosed by the red curve represents the safe set $\mathcal{X}$, and the blue points correspond to the system states visited over 10 time steps.}
    \label{fig:ex1_traj}
\end{figure}

\textbf{SDP formulation:}  
We solve the safety verification problem with $\epsilon_1 = 0.65$ and $\epsilon_1 = 0.75$, as defined in Definition~\ref{safe}, via solving the constraint \eqref{constraints1_safe}. The corresponding SDP over unknown polynomials $(v(x), s_0(x), s_1(x))$ is:
\begin{equation*}
\begin{cases}
1-\epsilon_1-v(x_0)\geq 0,\\
v(x)-\mathbb{E}_{\theta}[v(f(x,\theta))]+s_0(x) h(x) \in \sum[x],\\
v(x)-1-s_1(x)h(x)\in \sum[x],\\
v(x)\in \sum[x],\quad s_0(x),s_1(x)\in \sum[x].
\end{cases}
\end{equation*}

\textbf{SDP feasibility vs polynomial degree on ($v(x)$,$s_0(x)$, $s_1(x)$):}  
Table~\ref{tab:sdp_ex1} summarizes which degrees yield feasible SDPs.

\begin{table}[ht]
\caption{\centering SDP feasibility for Example~\ref{ex1} with $x_0=-0.8$\\
(\ding{52}: feasible; 
\ding{55}: infeasible)}
\label{tab:sdp_ex1}
\centering
\begin{tabular}{|c|c|c|}
\hline
Degree & $\epsilon_1=0.65$ & $\epsilon_1=0.75$ \\
\hline
6 & \ding{52} & \ding{55} \\
18 & \ding{52} & \ding{52} \\
\hline
\end{tabular}
\end{table}

\textbf{Extension to initial set (as indicated in Remark \ref{remark1}):}  
For initial set $\mathcal{X}_0=\{x\mid h_0(x)\le 0\}$ with $h_0(x)=(x-0.8)^2-0.01$,  the resulting SDP over unknown polynomials $(v(x),s_i(x),i=0,\ldots,2)$ is 
 \begin{equation*}
\begin{cases}
  1-\epsilon_1-v(x)+s_0(x)h_0(x)\in \sum[x], \\
       v(x)-\mathbb{E}_{\theta}[v(f(x,\theta))]+s_1(x) h(x) \in \sum[x],\\
        v(x)-1-s_2(x)h(x)\in \sum[x],\\
       v(x)\in \sum[x],s_0(x)\in \sum[x],\\
       s_1(x)\in \sum[x],s_2(x)\in \sum[x].
       \end{cases}
\end{equation*}
It is feasible for degree 18 but infeasible for degree 16, as reported in Table \ref{tab:sdp_ex1_set}.  We also compute an empirical estimate of the safety probability using a Monte Carlo method. Specifically, we draw $10^3$ initial states independently from the initial set $\mathcal{X}_0$ according to a uniform distribution, and for each initial state, we simulate $10^4$ trajectories over $10^4$ time steps. This procedure yields an estimated safety probability of $0.7521$.
\begin{table}[ht]
\caption{\centering SDP feasibility for Example~\ref{ex1} with the initial set $\mathcal{X}_0$\\
(\ding{52}: feasible; 
\ding{55}: infeasible)}
\label{tab:sdp_ex1_set}
\centering
\begin{tabular}{|c|c|c|c|}
\hline
Degree & $\epsilon_1=0.60$ &$\epsilon_1=0.65$\\
\hline
14 & \ding{52}& \ding{55}\\
22 & \ding{52}& \ding{52}\\
\hline
\end{tabular}
\end{table}

\end{example}

\begin{example}[Reach-Avoid Verification]
\label{ex2}
Consider the same system as Example~\ref{ex1}. The safe set is $\mathcal{X}=\{x\mid h(x)\leq 0\}$ with $h(x)=x^2-1$, the target set is $\mathcal{X}_r=\{x\mid g(x)\leq 0\}$ with $g(x)=x^2-0.01$, and the initial state is $x_0=-0.8$. We simulate $10^4$ trajectories over $10^4$ time steps. The estimated reach-avoid probability is $0.8240$. Figure~\ref{fig:ex2_traj} shows three example trajectories over 10 time steps.

\begin{figure}[htbp]
    \centering
    \includegraphics[width=0.45\textwidth]{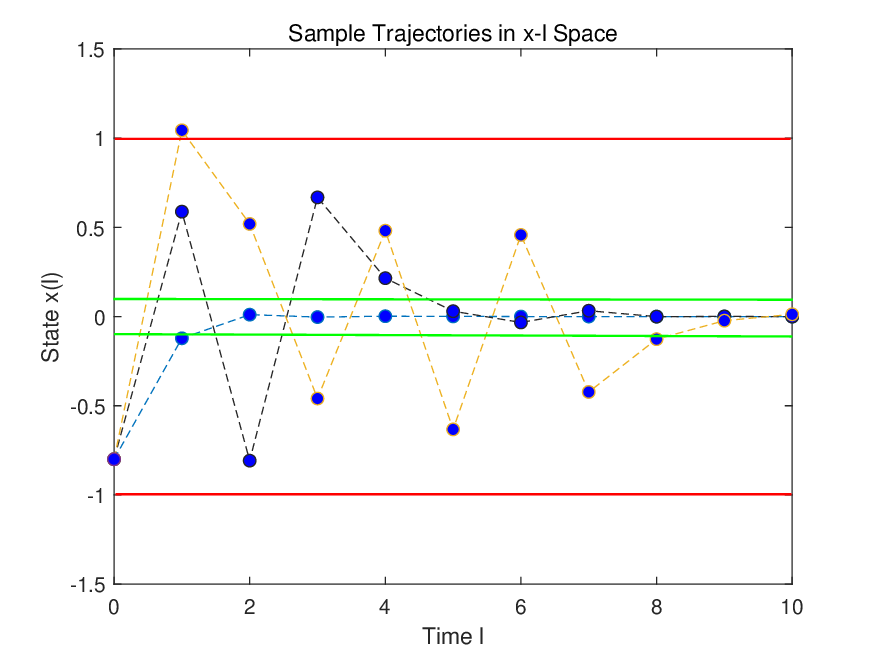} 
    \caption{Trajectories of the system in Example~\ref{ex2} starting from $x_0=-0.8$. The regions enclosed by the red and green curves represent the safe set $\mathcal{X}$ and the target set $\mathcal{X}_r$, respectively, while the blue points indicate the system states visited over 10 time steps.}
    \label{fig:ex2_traj}
\end{figure}

\textbf{SDP formulation:} We solve the reach-avoid verification problem with $\epsilon_2 = 0.65$ and $\epsilon_2 = 0.75$, as defined in Definition~\ref{reach-avoid}, by solving the constraints \eqref{constraints11} and \eqref{constraint_gamma0}. 
As proven in Proposition \ref{almost} (see Appendix), for any $x\in \mathcal{X}\setminus \mathcal{X}_r$, the trajectory will leave this set in finite time with probability one. Therefore, Assumption \ref{assump1} is satisfied.
Therefore, we can address the reach-avoid verification problem by solving constraint \eqref{constraints11}. The key difference between constraints \eqref{constraints11} and \eqref{constraint_gamma0} lies in the treatment of the term $\mathbb{E}_{\theta}[v(f(x,\theta))]$: in \eqref{constraints11}, this term is multiplied by $1$, whereas in \eqref{constraint_gamma0}, it is scaled by a discount factor $\gamma \in (0,1)$. As a result, their corresponding SDPs—formulated over the unknown polynomials ($v(x)$, $s_i(x)$, $i=0,\ldots,3$)—can be expressed in a unified form by treating $\gamma$ as a tunable parameter, as shown below.
Specifically, setting $\gamma=1$ recovers constraint \eqref{constraints11}, while choosing $\gamma \in (0,1)$ corresponds to constraint \eqref{constraint_gamma0}.
 \begin{equation*}
\begin{cases}
  v(x_0)-\epsilon_2\geq 0, \\
       \gamma \mathbb{E}_{\theta}[v(f(x,\theta))]-v(x)+s_0(x) h(x)-s_1(x) g(x) \in \sum[x],\\
        1-v(x)+s_2(x)g(x)\in \sum[x],\\
       -v(x)-s_3(x)h(x)\in \sum[x],\\
       s_0(x)\in \sum[x],s_1(x)\in \sum[x],\\
       s_2(x)\in \sum[x],s_3(x)\in \sum[x].
       \end{cases}
\end{equation*}

\textbf{SDP feasibility vs polynomial degree  ($v(x)$, $s_i(x)$, $i=0,\ldots,3$) and $\gamma$:}  
Table~\ref{tab:sdp_ex2} summarizes feasibility results.

\begin{table}[htbp]
\centering
\caption{\centering SDP feasibility for Example~\ref{ex2} with $\bm{x}_0=-0.8$\\
(\ding{52}: feasible; 
\ding{55}: infeasible)}
\label{tab:sdp_ex2}
\begin{tabular}{|c|c|c|c|c|}
\hline
Degree & $\epsilon_2$& $\gamma=1$ & $\gamma=0.999$ & $\gamma=0.99$ \\
\hline
6 & 0.65&  \ding{52} & \ding{52} & \ding{55} \\
6 & 0.75&  \ding{55} & \ding{55} & \ding{55} \\
18 & 0.65&  \ding{52} & \ding{52} & \ding{52} \\
18 & 0.75&  \ding{52} & \ding{52} & \ding{55} \\
\hline
\end{tabular}
\end{table}

\textbf{Extension to initial set (as indicated in Remark \ref{remark4}):}  
For $\mathcal{X}_0=\{x\mid h_0(x)\le 0\}$ with $h_0(x)=(x-0.8)^2-0.01$, the resulting SDP over unknown polynomials $(v(x),s_i(x),i=0,\ldots,4)$ is 
 \begin{equation*}
\begin{cases}
   v(x)-\epsilon_2+s_0(x)h_0(x)\in \sum[x], \\
       \gamma \mathbb{E}_{\theta}[v(f(x,\theta))]-v(x)+s_1(x) h(x)-s_2(x) g(x) \in \sum[x],\\
        1-v(x)+s_3(x)g(x)\in \sum[x],\\
       -v(x)-s_4(x)h(x)\in \sum[x],\\
       s_0(x)\in \sum[x],s_1(x)\in \sum[x],\\
       s_2(x)\in \sum[x],s_3(x)\in \sum[x], s_4(x)\in \sum[x].
       \end{cases}
\end{equation*}
Its feasibility summarized in Table \ref{tab:sdp_ex2_set}. We also compute an empirical estimate of the reach-avoid probability using a Monte Carlo method. Specifically, we draw $10^3$ initial states independently from the initial set $\mathcal{X}_0$ according to a uniform distribution, and for each initial state, we simulate $10^4$ trajectories over $10^4$ time steps. This procedure yields an estimated reach-avoid probability of $0.7510$.

\begin{table}[htbp]
\centering
\caption{\centering SDP feasibility for Example~\ref{ex2} with the initial set $\mathcal{X}_0$\\
(\ding{52}: feasible; 
\ding{55}: infeasible)}
\label{tab:sdp_ex2_set}
\begin{tabular}{|c|c|c|c|c|}
\hline
Degree & $\epsilon_2$ & $\gamma=1$ & $\gamma=0.999$ & $\gamma=0.99$ \\
\hline
20 & 0.65& \ding{55} & \ding{55} & \ding{55} \\
22 & 0.65& \ding{52} & \ding{52} & \ding{55} \\
\hline
\end{tabular}
\end{table}

\end{example}

The above examples demonstrate how the choice of polynomial degree and the factor $\gamma$ affect the feasibility of the safety and reach–avoid verification problems. Increasing the polynomial degree enhances the representational capacity of the barrier-like functions, which benefits both safety and reach–avoid verification by enabling the SDP to satisfy the associated conditions for larger tolerance parameters $\epsilon_1$ and $\epsilon_2$. In contrast, the discount factor $\gamma$, which is specific to the reach-avoid verification, influences the trade-off between conservatism and feasibility: values of $\gamma$ closer to 1 generally make the SDP more likely to be feasible under less conservative conditions.

\color{black}

\section{Conclusion}
\label{sec:con}
In this paper, we demonstrated necessary and sufficient barrier-like conditions for safety and reach-avoid verification of stochastic discrete-time systems over the infinite-time horizon. These conditions were constructed via relaxing Bellman equations. 

As indicated in Remark~\ref{remark5}, extending the result of Theorem~\ref{thm3} from a singleton initial state ${\bm{x}_0}$ to a general initial set $\mathcal{X}_0$ would require additional assumptions, such as uniform convergence properties. This extension will be investigated in future work.
Furthermore, we will develop efficient numerical methods to address the proposed barrier-like constraints for safety and reach–avoid verification of general nonlinear discrete-time stochastic systems. In addition, while infinite-time safety and reach-avoid verification methods provide rigorous guarantees for indefinite operational durations, they often impose stringent requirements that can be overly conservative, particularly in systems subject to stochastic disturbances such as additive Gaussian noise. In contrast, finite-time verification is more aligned with practical applications, where systems typically operate within bounded time horizons. Thus, finite-time verification presents a more practical approach for these systems. We will explore the necessary and sufficient barrier-like conditions for finite-time safety and reach-avoid verification in stochastic discrete-time systems.
\bibliographystyle{abbrv}
\bibliography{reference}
\section{Appendix}

\textbf{The proof of Proposition \ref{upper_condition_ra}}:

\begin{proof}
1) We first prove the ``only if'' part.

We prove via induction that for all $k\in \mathbb{N}$,
\[
\begin{split}
\zeta_k(\bm{x}) := &\mathbb{E}_{\pi}\left[ \sum_{i=0}^{k} \prod_{j=0}^{i-1} 1_{\mathcal{X}\setminus \mathcal{X}_r}(\bm{\phi}^{\bm{x}}_{\pi}(j)) \cdot 1_{\mathcal{X}_r}(\bm{\phi}^{\bm{x}}_{\pi}(i)) \right] + \mathbb{E}_{\pi}\left[ \prod_{j=0}^{k} 1_{\mathcal{X}\setminus \mathcal{X}_r}(\bm{\phi}^{\bm{x}}_{\pi}(j)) \cdot v(\bm{\phi}^{\bm{x}}_{\pi}(k+1)) \right]\\
\leq &v(\bm{x}).
\end{split}
\]

\textbf{Base Case ($k=0$):}
\[
\begin{split}
\zeta_0(\bm{x}) &= \mathbb{E}_{\pi}\left[1_{\mathcal{X}_r}(\bm{\phi}^{\bm{x}}_{\pi}(0))\right] + \mathbb{E}_{\pi}\left[1_{\mathcal{X}\setminus \mathcal{X}_r}(\bm{\phi}^{\bm{x}}_{\pi}(0)) v(\bm{\phi}^{\bm{x}}_{\pi}(1))\right] \\
&= 1_{\mathcal{X}_r}(\bm{x}) + 1_{\mathcal{X}\setminus \mathcal{X}_r}(\bm{x}) \mathbb{E}_{\bm{\theta}}[v(\bm{f}(\bm{x},\bm{\theta}))] \leq v(\bm{x}),
\end{split}
\]
where the first equality follows from the convention that the empty product equals 1, and the inequality follows from condition \eqref{constraints1}.

\textbf{Inductive Step:} Assume $v(\bm{x}) \geq \zeta_k(\bm{x})$ for some $k \geq 0$. Then:
\begin{align*}
\zeta_{k+1}(\bm{x}) 
&= \zeta_k(\bm{x}) 
- \mathbb{E}_{\pi} \Bigg[ 
\prod_{j=0}^{k} 1_{\mathcal{X}\setminus \mathcal{X}_r}\big(\bm{\phi}^{\bm{x}}_{\pi}(j)\big) \,
v\big(\bm{\phi}^{\bm{x}}_{\pi}(k+1)\big) 
\Bigg] \\
&+ \mathbb{E}_{\pi} \Bigg[ 
\prod_{j=0}^{k} 1_{\mathcal{X}\setminus \mathcal{X}_r}\big(\bm{\phi}^{\bm{x}}_{\pi}(j)\big) \Bigg( 
1_{\mathcal{X}_r}\big(\bm{\phi}^{\bm{x}}_{\pi}(k+1)\big) + 1_{\mathcal{X}\setminus \mathcal{X}_r}\big(\bm{\phi}^{\bm{x}}_{\pi}(k+1)\big) \,
\mathbb{E}_{\bm{\theta}} \bigg[
v\big(\bm{\phi}^{\bm{x}}_{\pi}(k+2)\big)
\bigg] \Bigg) 
\Bigg].
\end{align*}

Using condition \eqref{constraints1} at state $\bm{\phi}^{\bm{x}}_{\pi}(k+1)$:
\[
\begin{split}
v(\bm{\phi}^{\bm{x}}_{\pi}(k+1)) \geq &1_{\mathcal{X}_r}(\bm{\phi}^{\bm{x}}_{\pi}(k+1)) + 1_{\mathcal{X}\setminus \mathcal{X}_r}(\bm{\phi}^{\bm{x}}_{\pi}(k+1)) \mathbb{E}_{\bm{\theta}}[v(\bm{\phi}^{\bm{x}}_{\pi}(k+2))],
\end{split}
\]
we have $\zeta_{k+1}(\bm{x}) \leq \zeta_k(\bm{x}) \leq v(\bm{x})$.

By induction, $v(\bm{x}) \geq \zeta_k(\bm{x})$ for all $k \in \mathbb{N}$.  Since $\zeta_k(\bm{x})\geq 0$ for all $k \in \mathbb{N}$, $\lim_{k\to\infty} \zeta_k(\bm{x})$ exists. Taking $k \to \infty$, we have
\[
\begin{split}
\lim_{k\to\infty} \zeta_k(\bm{x}) &= \mathbb{E}_{\pi}\left[ \sum_{i=0}^{\infty} \prod_{j=0}^{i-1} 1_{\mathcal{X}\setminus \mathcal{X}_r}(\bm{\phi}^{\bm{x}}_{\pi}(j)) \cdot 1_{\mathcal{X}_r}(\bm{\phi}^{\bm{x}}_{\pi}(i)) \right] + \lim_{k\to\infty} \mathbb{E}_{\pi}\left[ \prod_{j=0}^{k} 1_{\mathcal{X}\setminus \mathcal{X}_r}(\bm{\phi}^{\bm{x}}_{\pi}(j)) v(\bm{\phi}^{\bm{x}}_{\pi}(k+1)) \right]\\
&\geq \mathbb{E}_{\pi}\left[ \sum_{i=0}^{\infty} \prod_{j=0}^{i-1} 1_{\mathcal{X}\setminus \mathcal{X}_r}(\bm{\phi}^{\bm{x}}_{\pi}(j)) \cdot 1_{\mathcal{X}_r}(\bm{\phi}^{\bm{x}}_{\pi}(i)) \right]\\
&=V(\bm{x}).
\end{split}
\]
 Thus, $v(\bm{x}) \geq V(\bm{x})$.

Therefore, according to Lemma \ref{equal}, we have \[\mathbb{P}_{\pi}(S'_{\bm{x}_0})=V(\bm{x}_0)\leq \epsilon'_1.\]
2) We will prove the ``if'' part.

If  $\mathbb{P}_{\pi}(S'_{\bm{x}_0})\leq \epsilon'_1$, we have $V(\bm{x}_0)\leq \epsilon'_1$ according to Lemma \ref{equal}, where $V(\cdot):\mathbb{R}^n \rightarrow \mathbb{R}$ is the value function in \eqref{value}. Moreover, according to Proposition \ref{pro_bellman}, $V(\bm{x})$ satisfies $V(\bm{x})= \mathbb{E}_{\bm{\theta}}[V(\bm{f}(\bm{x},\bm{\theta}))], \forall \bm{x}\in \mathcal{X}\setminus \mathcal{X}_r$, $V(\bm{x})=1, \forall \bm{x}\in \mathcal{X}_r$, and $V(\bm{x})= 0, \forall \bm{x}\in \mathbb{R}^n\setminus \mathcal{X}$. Consequently, $V(\bm{x})$ satisfies \eqref{constraints1}. 
\end{proof}

\begin{proposition}
\label{almost}
Let $\mathcal X=[-1,1]$ and $\mathcal X_r=[-0.1,0.1]$. Starting from any $x(0)\in \mathcal X\setminus\mathcal X_r$, the system \[ x(l+1)=(-0.5+\theta(l))\,x(l), \] where $\{\theta(l)\}_{l\ge0}$ are independent and identically distributed, each drawn uniformly from the interval $[-1,1]$, leaves $\mathcal X\setminus\mathcal X_r$ in finite time almost surely.   
\end{proposition}
\begin{proof}
Step 1 System Reformulation

We have $x(l+1) = (-0.5 + \theta(l)) x(l)$ with $\theta(l)$ i.i.d. uniform on $[-1,1]$. Then $A(l) = -0.5 + \theta(l)$ is uniform on $[-1.5, 0.5]$. Moreover, $|x(l)| = |x(0)| \prod_{k=0}^{l-1} |A(k)|$.

Step 2. Logarithmic Transformation

Let \(y(l) = \log |x(l)|\). Then:
\[
y(l) = \log |x(0)| + \sum_{k=0}^{l-1} Y_k, \quad \text{where } Y_k = \log |A(k)|.
\]
The sequence \(\{Y_k\}\) is i.i.d.

Step 3. Lyapunov Exponent

Computing the expectation:
\[
\begin{split}
\mathbb{E}[Y_k] &= \frac{1}{2} \int_{-1.5}^{0.5} \log |a| \, da= \frac{1}{2}\big(\int_{0}^{1.5} \log u \, du + \int_{0}^{0.5} \log u \, du\big)
\end{split}
\]
and using \(\int \log u \, du = u \log u - u\), we get:
\[
\begin{split}
\mathbb{E}[Y_k] &= \frac{1}{2} \left( (1.5 \log 1.5 - 1.5) + (0.5 \log 0.5 - 0.5) \right) \\
&= 0.75 \log 1.5 + 0.25 \log 0.5 - 1.
\end{split}
\]
Numerically, \(\log 1.5 \approx 0.405\), \(\log 0.5 \approx -0.693\), so:
\[
\mathbb{E}[Y_k] \approx 0.304 - 0.173 - 1 = -0.869 < 0.
\]
Let \(\lambda = \mathbb{E}[Y_k] < 0\).

Step 4. Almost Sure Convergence

By the strong law of large numbers, we have
\[
\frac{1}{l} \sum_{k=0}^{l-1} Y_k \to \lambda \quad \text{almost surely}.
\]
Hence, \(\sum_{k=0}^{l-1} Y_k \to -\infty\) almost surely as $l \to +\infty$, so:
\[
y(l) = \log |x(0)| + \sum_{k=0}^{l-1} Y_k \to -\infty \quad \text{almost surely}
\]
as $l \to +\infty$. Therefore, \(|x(l)| \to 0\) almost surely as $l \to +\infty$.

Step 5. Exit from the Set $\mathcal{X}\setminus \mathcal{X}_r$

For any \(x(0) \in \mathcal{X}\setminus \mathcal{X}_r = [-1,1] \setminus [-0.1,0.1]\), we have \(|x(0)| \geq 0.1\). Since \(|x(l)| \to 0\) almost surely as $l \to +\infty$, there exists almost surely a finite time \(N\) such that \(|x(N)| < 0.1\), i.e., \(x(N) \notin \mathcal{X}\setminus \mathcal{X}_r\).

The system eventually exits the set \(\mathcal{X}\setminus \mathcal{X}_r\) from every initial state in \(\mathcal{X}\setminus \mathcal{X}_r\) almost surely.

This completes the proof.
\end{proof}
\end{document}